\documentclass[lettersize,journal]{IEEEtran}
\usepackage{amsmath,amsfonts}
\usepackage{algorithmic}
\usepackage{algorithm}
\usepackage{array}
\usepackage{textcomp}
\usepackage{stfloats}
\usepackage{url}
\usepackage{verbatim}
\usepackage{graphicx}
\usepackage{cite}
\usepackage{amssymb}
\usepackage{amsthm}
\usepackage{mathrsfs}
\usepackage{siunitx}
\usepackage{multirow}
\usepackage{xcolor}
\usepackage{bm}
\usepackage{booktabs}

\usepackage{subcaption}
\newtheorem{lemma}{Lemma}

\newtheoremstyle{ieeeremark}%
  {\topsep}
  {\topsep}
  {\normalfont}
  {0pt}
  {\bfseries}
  {:}
  {0.5em}
  {}

\theoremstyle{ieeeremark}
\newtheorem{remark}{Remark}

\begin{document}

\title{MCRB and MSE Analysis for Parameter Estimation in AFDM-ISAC Systems}

\author{Tianyao Ma, Aimin Tang,~\IEEEmembership{Senior Member,~IEEE,} Yin Xu,~\IEEEmembership{Senior Member,~IEEE,} Qu Luo,~\IEEEmembership{Member,~IEEE,} Dazhi He,~\IEEEmembership{Senior Member,~IEEE,} Wenjun Zhang,~\IEEEmembership{Fellow,~IEEE}
\thanks{Tianyao Ma, Aimin Tang, Yin Xu, Dazhi He, and Wenjun Zhang are from Shanghai Jiao Tong University (e-mail: mty0710, tangaiming, xuyin, hedazhi, zhangwenjun@sjtu.edu.cn). \emph{(Corresponding author: Yin Xu.)}}
\thanks{Qu Luo is from University of Surrey (e-mail: q.u.luo@surrey.ac.uk).}}



\maketitle

\begin{abstract}
Affine frequency division multiplexing (AFDM) is a promising waveform for integrated sensing and communication (ISAC). In AFDM systems, the complex gains, delays, and Doppler shifts are commonly estimated from the AFDM symbols carrying pilots and data simultaneously. In practice, however, the unknown data symbols and data-pilot coupling interference may render the estimator mismatched to the true signal model.
In this paper, we systematically characterize the parameter-estimation performance of AFDM-ISAC systems under practical model misspecification. The main contributions are threefold. First, we extend the Cram\'er-Rao bound (CRB) for a general observation model that treats the data symbols as unknown, which generalizes existing AFDM CRB analyses and serves as the matched benchmark for the subsequent analysis. Second, we identify two practical sources of misspecification, namely a covariance mismatch caused by insufficient pilot-data isolation and a combined covariance-and-mean mismatch caused by sequential single-target estimation, and derive the corresponding misspecified CRB (MCRB). Third, we characterize the pseudotrue parameters under different levels of prior knowledge, analyze the resulting estimation bias, and establish a lower bound (LB) on the mean square error (MSE). Simulation results validate the derived bounds and show that, under model misspecification, the CRB is overly optimistic while the MCRB and LB faithfully characterize the achievable accuracy. The comparison further reveals how these bounds vary with the pilot length and pilot power, providing useful guidance for pilot configuration.
\end{abstract}

\begin{IEEEkeywords}
Affine frequency division multiplexing (AFDM), integrated sensing and communication (ISAC), misspecified Cram\'er-Rao bound (MCRB), parameter estimation, model misspecification.
\end{IEEEkeywords}

\section{Introduction}
\IEEEPARstart{T}{he} proliferation of autonomous vehicles, unmanned aerial vehicles (UAVs), and smart transportation systems has driven a surging demand for integrated sensing and communication (ISAC) in next-generation communication systems \cite{liu2020joint,tang2021self,liu2022fundamental}.
These high-mobility applications require wireless signals to support reliable data transmission while simultaneously extracting environmental information from the same spectrum, hardware, and signal processing resources \cite{liu2022isac6g}. 
Meanwhile, the underlying channels are typically doubly dispersive, exhibiting both time delays due to multipath propagation and Doppler shifts induced by high-speed mobility, which poses fundamental challenges for both communication and sensing.
Legacy waveforms such as orthogonal frequency division multiplexing (OFDM) suffer from severe inter-carrier interference (ICI) under such conditions, degrading both communication reliability and sensing accuracy \cite{bemani2024integrated}. 
To address this issue, emerging waveforms such as orthogonal time frequency space (OTFS) \cite{raviteja2018otfs}, orthogonal chirp division multiplexing (OCDM) \cite{ouyang2016ocdm}, and affine frequency division multiplexing (AFDM) \cite{AFDM} have been investigated. Among them, AFDM is a one-dimensional chirp-based multicarrier waveform operating in the discrete affine Fourier transform (DAFT) domain. Owing to its relatively low implementation complexity and ability to achieve full diversity over doubly dispersive channels, AFDM has recently attracted considerable attention as a promising waveform for high-mobility communications and ISAC \cite{bemani2024integrated,AFDM,afdmScma,afdmStandardization2026,chirpSurvey2026}.


In AFDM-based ISAC systems, sensing and channel estimation share a similar objective, i.e., estimating the complex channel gains, time delays, and Doppler shifts of sensing targets or propagation paths.
Recent AFDM studies have addressed several complementary aspects of this parameter estimation problem. For communication over doubly dispersive channels, pilot-aided schemes have been developed to estimate the sparse channel parameters \cite{AFDM,EP,SP,sctPilot,benzine2026models}. For AFDM-ISAC, sensing algorithms have been studied for delay-Doppler estimation \cite{bemani2024integrated,ni2025twc,luo2025novel}. In addition, AFDM waveform and parameter optimization has been investigated to balance communication and sensing performance \cite{bao2024performance,cao2025agile}. 
These studies highlight the need to quantify the achievable estimation accuracy.
The Cram\'er-Rao bound (CRB) provides a fundamental lower bound (LB) on the mean squared error (MSE) of any unbiased estimator \cite{slepian1954estimation,bangs1971array}.
For OFDM-ISAC systems, CRB-based analyses have been extensively developed to characterize sensing-performance limits and guide system design \cite{liu2022fundamental,wang2026device,mura2025optimized,liu2025carrier,gupta2025data}.
However, due to the differences in the waveforms and pilot structures, the CRB results for OFDM cannot be directly applied to AFDM systems.

Driven by this need, several recent works have investigated the CRB for AFDM. In \cite{zhang2026CRBpilot}, the CRB for joint estimation of channel gain, delay, and Doppler was derived under a single-target monostatic sensing model and used to guide pilot design. The CRB was also employed in \cite{bao2024performance} to evaluate the communication-sensing performance tradeoff achieved by adjusting the AFDM parameters and was further adopted as the objective of a waveform optimization problem in \cite{cao2025agile}, where the delay and Doppler CRBs are minimized over the DAFT chirp parameters. Unlike the above works \cite{bao2024performance,cao2025agile,zhang2026CRBpilot}, which consider a single target, the authors in \cite{wang2025afdmLEO} considered multiple dominant paths in LEO inter-satellite links and incorporated the trace of the CRB matrix as a regularization term in the design of minimum mean square error (MMSE) pre-equalizers.
However, all the aforementioned studies assume that the entire transmitted signal is perfectly known at the receiver, with randomness arising only from additive noise. This assumption excludes the case where data symbols coexist with pilots but remain unknown to the sensing or communication receiver, thereby introducing additional randomness that changes the covariance structure of the received observation.

More importantly, these CRB analyses also presume that the estimator is perfectly matched to the true signal model. In practice, however, AFDM-based parameter estimation is often subject to model mismatch. Existing schemes commonly use guard symbols to separate the pilot observation from surrounding data symbols in the DAFT domain \cite{AFDM,EP,benzine2026models}. 
When the guards are insufficient, the data-induced component remains in the observation covariance, whereas conventional matched-filtering estimators \cite{bemani2024integrated,li2026matched,sctPilot} implicitly assume a scaled-identity covariance and thus disregard this component, leading to a covariance mismatch. 
In sequential single-target estimation with an unknown number of targets \cite{AFDM,li2026matched,sctPilot}, the targets that have not been estimated are absent from the assumed signal mean, introducing an additional mean mismatch. In both cases, the assumed probability density function (PDF) no longer matches the true distribution of the received signal, which is referred to as model misspecification.

To address the model misspecification, a generalized analytical framework has been established \cite{vuong1986cramer, richmond2015parameter, MCRB}.
Within this framework, the CRB is no longer a valid LB on performance for two reasons. 
First, since no parameter value renders the assumed PDF equal to the true one, no true parameter exists within the assumed model in general. The estimator instead converges to the so-called pseudotrue parameter, which minimizes the Kullback-Leibler divergence (KLD) between the assumed and true PDFs.
Second, the gap between this pseudotrue parameter and the true one introduces an estimation bias. As a result, the achievable performance is more accurately characterized by the misspecified CRB (MCRB) \cite{vuong1986cramer,richmond2015parameter}. The MCRB lower-bounds the error covariance of any misspecified-unbiased estimator relative to the pseudotrue parameter, while the overall MSE referred to the true parameter is lower-bounded by the sum of the MCRB and the squared bias \cite{MCRB}.
These bounds reduce to the classical CRB when the model is correctly specified. 
The MCRB framework has been applied to a variety of parameter-estimation problems with different sources of model misspecification. In OFDM-based ranging and localization, the ranging bias and accuracy degradation caused by unmodeled multipath were quantified in \cite{lucas2026ofdm}, while the impact of hardware impairments on localization accuracy was characterized in \cite{chen2024ofdmHWI}. For near-field reconfigurable intelligent surface (RIS)-aided localization, the performance loss incurred when adopting a simplified far-field model was analyzed in \cite{sun2026nearfield}. Beyond localization, the MCRB under ignored multipath components was derived for direction-of-arrival estimation in automotive multiple-input multiple-output (MIMO) radar in \cite{moshe2023mcrb}, while channel-order misspecification was analyzed for blind channel estimation in \cite{thanh2021blind}.
Nevertheless, an MCRB analysis tailored to AFDM-ISAC systems remains absent.

Motivated by this gap, in this paper we systematically investigate the parameter-estimation performance of AFDM-ISAC systems under both matched and misspecified conditions. We derive the relevant performance bounds, characterize the pseudotrue parameters, and validate and compare these bounds through simulations. The main contributions of this work are summarized as follows:
\begin{itemize}
\item We extend the CRB for a general AFDM-ISAC observation model in which the data symbols transmitted together with the pilots are treated as unknown. This generalizes existing AFDM CRB analyses, which regard the transmitted data as deterministic and known at the receiver, and serves as the matched benchmark for the subsequent mismatched analysis.

\item We identify and analyze two practical sources of model misspecification in AFDM parameter estimation, namely a covariance mismatch caused by insufficient pilot-data isolation, and a combined covariance-and-mean mismatch caused by sequential single-target estimation. For each case, we clarify the mechanism, formulate the observation model and its PDF, and derive the corresponding MCRB.

\item We characterize the pseudotrue parameters and investigate the MSE under model misspecification. In particular, for the combined covariance-and-mean mismatch, we derive the pseudotrue parameters under different levels of prior knowledge of the unknown parameters. The resulting bias, together with the MCRB, yields an LB on the MSE with respect to the true parameters, thereby characterizing the attainable error floor.

\item We validate the derived CRB by comparing it with the MSE of matched estimators under both unknown- and known-data conditions. We further validate the derived MCRB and LB on the MSE of the misspecified estimator, and compare the CRB, MCRB, and LB under different pilot configurations. The results reveal the impact of pilot length and pilot power, providing guidelines for pilot configuration under practical mismatched estimation.
\end{itemize}

The remainder of this paper is organized as follows. The system model is presented in Section \ref{sec:system}. The CRB and MCRB are derived in Section \ref{sec:CRBandMCRB}. The pseudotrue parameters, bias, and MSE are analyzed in Section \ref{sec:MSE}, followed by the simulation results presented in Section \ref{sec:simu}. Finally, this paper is concluded in Section \ref{sec:conclusion}.

\emph{Notation:} Italic, bold lowercase, and bold uppercase letters denote scalars, vectors, and matrices, respectively. $\mathbb{R}$ denotes the set of real numbers. $\lfloor\cdot\rfloor$, $\Re\{\cdot\}$, and $\Im\{\cdot\}$ denote the floor, real-part, and imaginary-part operations. $||\cdot||_0$ and $||\cdot||$ denote the $\ell_0$- and $\ell_2$-norms, respectively. $(\cdot)^\text{T}$, $(\cdot)^\text{H}$, and $(\cdot)^{-1}$ denote transpose, Hermitian transpose, and matrix inversion, respectively. $[\mathbf{A}]_{m,n}$ and $[\mathbf{A}]_{a:b,c:d}$ denote the $(m,n)$-th entry and the sub-block formed by rows $a$ through $b$ and columns $c$ through $d$ of $\mathbf{A}$, respectively. $\text{diag}(a_n,n=0,1,...,N-1)$ denotes a diagonal matrix with diagonal entries $a_0,\ldots,a_{N-1}$, and $\mathrm{tr}(\cdot)$ denotes the trace. 
$\mathbf{I}_N$ and $\mathbf{0}$ denote the $N\times N$ identity matrix and an all-zero vector or matrix, respectively. $\mathcal{CN}(\bm{\mu},\mathbf{C})$ denotes the complex Gaussian distribution with mean $\bm{\mu}$ and covariance $\mathbf{C}$. $\mathbb{E}_p[\cdot]$ and $\mathrm{Cov}[\cdot]$ denote expectation under PDF $p$ and covariance, respectively.

\section{System Model} \label{sec:system}

\subsection{AFDM Modulation and Demodulation} \label{sec:AFDM_mod}
In AFDM systems, the DAFT-domain input signal $\mathbf{x} \in \mathbb{A}^N$, where $\mathbb{A}$ denotes the modulation alphabet, is mapped onto $N$ orthogonal chirp subcarriers. Specifically, by applying the inverse DAFT (IDAFT), the time-domain transmitted signal is given by
\begin{align}
    \mathbf{s} = \mathbf{A}^\text{H}\mathbf{x} = \mathbf{\Lambda}_{c_1}^\text{H} \mathbf{F}^\text{H} \mathbf{\Lambda}_{c_2}^\text{H} \mathbf{x},
    \label{AFDM_mod}
\end{align}
where $\mathbf{A} \triangleq \mathbf{\Lambda}_{c_2} \mathbf{F} \mathbf{\Lambda}_{c_1}$ denotes the DAFT matrix, $\mathbf{\Lambda}_{c} \triangleq \text{diag}(e^{-j2\pi cn^2}, n=0,1,...,N-1)$ with the AFDM modulation parameters $c \in \{c_1,c_2\}$ \cite{AFDM}. $\mathbf{F}$ denotes the discrete Fourier transform (DFT) matrix with its $(m,n)$-th entry given by $[\mathbf{F}]_{m,n} = \frac{1}{\sqrt{N}}e^{-j2\pi mn/N}$, $m,n=0,1,\ldots,N-1$.

The time-domain signal is transmitted after a chirp-periodic prefix (CPP)~\cite{AFDM} is appended, and is then reflected by the targets. At the receiver, after the CPP is removed, demodulation is performed as
\begin{align}
    \mathbf{y} = \mathbf{A}\mathbf{r} = \mathbf{\Lambda}_{c_2} \mathbf{F} \mathbf{\Lambda}_{c_1} \mathbf{r},
    \label{AFDM_demod}
\end{align}
where $\mathbf{y}$ and $\mathbf{r}$ are the received signals in the DAFT and discrete time domain, respectively.

\subsection{AFDM Input-Output Relationship} \label{sec:input-output}
Denote $P$ by the number of  sensing targets, and let $h_i$, $k_i$, and $l_i$ be the complex gain, normalized Doppler shift, and normalized integer delay associated with the $i$-th target, respectively, the received discrete-time-domain signal can be expressed as
\begin{align}
    \mathbf{r} = \mathbf{H}_\text{t}\mathbf{s} + \mathbf{w},
    \label{rs_TD}
\end{align}
where $\mathbf{H}_\text{t} \triangleq \sum_{i=1}^P h_i \mathbf{\Gamma}_{\text{CPP}_i} \mathbf{\Delta}_{k_i} \mathbf{\Pi}^{l_i}$ \cite{AFDM}. 
$\mathbf{w}\sim \mathcal{CN}(0,\sigma_\text{n}^2\mathbf{I}_N)$ is the complex Gaussian noise vector with the power of $\sigma_\text{n}^2$. 
$\mathbf{\Gamma}_{\text{CPP}_i} = \mathbf{I}_N$, i.e., an $N\times N$ identity matrix, when $2Nc_1$ is integer and $N$ is even. $\mathbf{\Delta}_{k_i} \triangleq \text{diag}(e^{-j2\pi{k_i}n/N}, n=0,1,...,N-1)$, and $\mathbf{\Pi}$ is the forward cyclic-shift matrix, as defined in \cite[Eq. (25)]{AFDM}.

The complex gains, normalized Doppler shifts, and normalized delays of the targets constitute the sensing parameters of interest in this work. 
For the $i$-th target, $k_i \in [-k_\text{max}, k_\text{max}]$ is normalized with the subcarrier spacing as $k_i = \nu_i/\Delta f$, where $k_\text{max} \triangleq \nu_\text{max}/\Delta f$ is the maximum normalized Doppler shift, $\nu_\text{max}$ is the maximum Doppler shift in \SI{}{Hz}, $\Delta f$ is the subcarrier spacing, and $\nu_i=2v_if_c/c$ denotes the Doppler shift in \SI{}{Hz} induced by the target's radial velocity $v_i$, with $f_c$ and $c$ denoting the carrier frequency and the speed of light, respectively.
Let $\tau_\text{max}$ denote the maximum delay in seconds and $l_\text{max}\triangleq \tau_\text{max}/T_\text{s}$ the maximum normalized delay. Similarly, $l_i$ is normalized as $l_i = \tau_i/T_\text{s}$ and bounded by $[0, l_\text{max}]$, where $\tau_i = 2R_i/c$ is the round-trip propagation delay in seconds corresponding to the target's range $R_i$ and $T_\text{s}$ represents the sampling time.
The complex gain $h_i$ of the $i$-th target subsumes its radar cross-section (RCS), round-trip path loss, and phase response. 

Combining (\ref{AFDM_mod}), (\ref{AFDM_demod}) and (\ref{rs_TD}), and denoting the effective channel matrix as $\mathbf{H}_\text{eff}$, the input-output relation of AFDM in the DAFT domain is given by
\begin{align}
   \mathbf{y} = \mathbf{H}_\text{eff}\mathbf{x} + \mathbf{\Tilde{w}},
   \label{yx_DAFT}
\end{align}
where $\mathbf{H}_\text{eff} \triangleq \sum_{i=1}^P h_i \mathbf{A} \mathbf{\Gamma}_{\text{CPP}_i} \mathbf{\Delta}_{k_i} \mathbf{\Pi}^{l_i} \mathbf{A}^\text{H}$ \cite{bemani2024integrated} and 
$\mathbf{\Tilde{w}}\sim \mathcal{CN}(0,\sigma_\text{n}^2\mathbf{I}_N)$ is the complex Gaussian noise vector in the DAFT domain. 

The chirp rate $c_1$ should be tuned to match the delay and Doppler characteristics of the channel. In general, to attain a large diversity gain and following existing works \cite{AFDM}, we choose $c_1 = \frac{2(k_\text{max}+k_\nu)+1}{2N}$, where $k_\nu$ accounts for the fractional-Doppler case and is set to zero for the integer-Doppler case. The choice of $c_2$ is relatively loose, as it can be any irrational number \cite{AFDM}.
With $c_1$ and $N$ properly chosen, $\mathbf{\Gamma}_{\text{CPP}_i}$ reduces to the identity matrix \cite{AFDM} and is omitted hereafter for brevity.

\subsection{Parameterized Observation Signal}
\begin{figure}[tb]
    \centering
    \includegraphics[width=0.9\columnwidth]{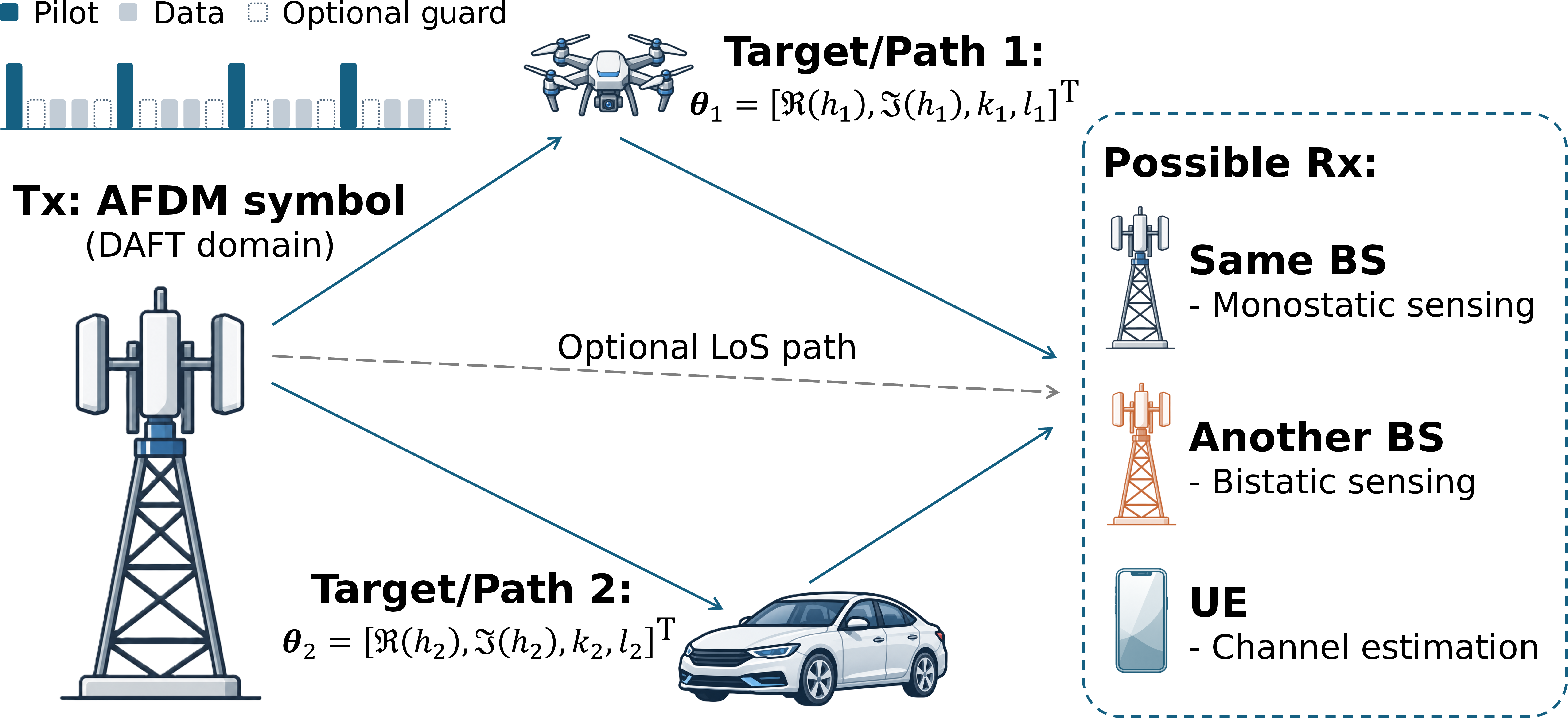}
    \caption{General AFDM parameter-estimation scenario with different types of receivers.}
    \label{fig:scenario}
\end{figure}

The general parameter-estimation scenario considered in this work is illustrated in Fig.~\ref{fig:scenario}. The receiver can be the transmitting base station (BS), another BS, or a user equipment (UE), corresponding to the monostatic sensing, bistatic sensing, and channel estimation scenarios, respectively. 
To facilitate parameter estimation, an embedded-pilot AFDM symbol structure is considered. In this paper, we consider a generalized case that covers both designs with and without guard subcarriers~\cite{bemani2024integrated,EP,sctPilot}.
Hence, the transmitted signal can be written as
\begin{align}
    \mathbf{x} = \mathbf{x}_\text{p} + \mathbf{x}_\text{d},
    \label{xp_xd}
\end{align}
where $\mathbf{x}_\text{p}$ and $\mathbf{x}_\text{d}$ denote the pilot and data vectors in the DAFT domain, respectively. If no guard intervals are inserted, then it holds that $||\mathbf{x}_\text{p}||_0 + ||\mathbf{x}_\text{d}||_0 = N$. Otherwise, $||\mathbf{x}_\text{p}||_0 + ||\mathbf{x}_\text{d}||_0 < N$ holds.
The data vector $\mathbf{x}_\text{d}$ has covariance $\mathrm{Cov}[\mathbf{x}_\text{d}] = \text{diag}(\sigma_\text{d}^2d_n,n=0,1,...,N-1)$, where $\sigma_\text{d}^2$ is the data power and $d_n \in \{0,1\}$ indicates whether the $n$-th subcarrier carries a data symbol. 

Substituting (\ref{xp_xd}) into (\ref{yx_DAFT}), the received signal in the DAFT domain, which serves as the observation signal for parameter estimation, is given by
\begin{align}
    \mathbf{y} = \sum_{i=1}^P ( \mathbf{y}_\text{p}(\bm{\theta}_i) + \mathbf{y}_\text{d}(\bm{\theta}_i) ) + \mathbf{\Tilde{w}}, \label{yp_yd}
\end{align}
where, $\bm{\theta}_i=[\Re(h_i), \Im(h_i), k_i, l_i]^\text{T}\in\bm{\Theta}_i$, $\mathbf{y}_\text{p}(\bm{\theta}_i) = h_i \mathbf{A} \mathbf{\Delta}_{k_i} \mathbf{\Pi}^{l_i} \mathbf{A}^\text{H} \mathbf{x}_\text{p}$, and $\mathbf{y}_\text{d}(\bm{\theta}_i) = h_i \mathbf{A} \mathbf{\Delta}_{k_i} \mathbf{\Pi}^{l_i} \mathbf{A}^\text{H} \mathbf{x}_\text{d}$ denote the parameter set to be estimated of the $i$-th target, the received pilot and data parts, respectively.
The feasible parameter space of each target is defined as $\bm{\Theta}_i\triangleq\bm{\Theta}_h\times\bm{\Theta}_k\times\bm{\Theta}_l$ with $\bm{\Theta}_h=\mathbb{R}^2$ and $\bm{\Theta}_k$ and $\bm{\Theta}_l$
denoting the feasible Doppler and delay domains, respectively.

\begin{remark}
    The signal model provided in this section is mathematically equivalent to that of AFDM channel estimation \cite{AFDM,bemani2024integrated}, where each target corresponds to a propagation path with the same parameters of channel gain, delay and Doppler shift. Therefore, the subsequent analysis directly applies to channel estimation scenarios. In the following sections, we focus on the ISAC scenario and use the term ``targets'' accordingly.
\end{remark}

\section{CRB and MCRB Analysis} \label{sec:CRBandMCRB}
This section develops the theoretical performance bounds for AFDM-ISAC parameter estimation. The CRB under a correctly specified statistical model is first established. The MCRB framework is then introduced for model mismatch and specialized to two practical mismatch cases.

\subsection{CRB for AFDM Parameter Estimation} \label{subsec:CRB}
The observation $\mathbf{y}$ in (\ref{yp_yd}) has a PDF that depends on the parameters to be estimated in general. 
If the PDF assumed by an estimator coincides with the actual one, the statistical model is said to be matched, or perfectly specified. Otherwise, it is referred to as mismatched, or misspecified.

Under perfect specification, all $P$ targets are estimated jointly and the true covariance structure is correctly accounted for at the receiver. The pilot part $\mathbf{x}_\text{p}$ and the effective channel matrix $\mathbf{H}_\text{eff}$ are deterministic, while the data part $\mathbf{x}
_\text{d}$ and the noise $\mathbf{\Tilde{w}}$ are treated as zero-mean random variables. 
In the presence of fractional Doppler and delay, each row of $\mathbf{H}_{\text{eff}}$ contains a sufficiently large number of non-zero elements, so that each component of $\mathbf{H}_\text{eff}\mathbf{x}_\text{d}$ can be expressed as a weighted sum of a large number of independent data symbols.
Invoking the central limit theorem (CLT), the elements of $\mathbf{y}$ can therefore be approximated as Gaussian distributed, and the true PDF of $\mathbf{y}$ is given by $p_Y(\mathbf{y}) = \mathcal{CN}(\bm{\mu}(\bm{\bar{\theta}}),\mathbf{C}(\bm{\bar{\theta}}))$, where $\bm{\bar{\theta}}=[\bm{\bar{\theta}}_1^\text{T},...,\bm{\bar{\theta}}_P^\text{T}]^\text{T}$ denotes the true channel parameters. The mean and covariance are respectively given by
\begin{align}
    \bm{\mu}(\bm{\bar{\theta}}) = \mathbb{E}_p[\mathbf{y}]=\sum_{i=1}^P \mathbf{y}_\text{p}(\bm{\bar{\theta}}_i) = \mathbf{H}_\text{eff}\mathbf{x}_\text{p} \label{E_true}
\end{align}
and
\begin{align}
    \mathbf{C} (\bm{\bar{\theta}})
    &= \mathrm{Cov}[\sum_{i=1}^P \mathbf{y}_\text{d}(\bm{\bar{\theta}}_i)] + \mathrm{Cov}[\mathbf{\Tilde{w}}] \notag \\
    &=\mathbf{H}_\text{eff} \mathrm{Cov}[\mathbf{x}_\text{d}] \mathbf{H}_\text{eff}^\text{H} + \sigma_{\text{n}}^2\mathbf{I}, \label{Cov_true}
\end{align}
where $\sigma_{\text{n}}^2$ denotes the power of the additive white Gaussian noise (AWGN). In general, $\mathbf{C}(\bm{\bar{\theta}})$ is non-diagonal, as the data-induced term $\mathbf{H}_\text{eff} \mathrm{Cov}[\mathbf{x}_\text{d}] \mathbf{H}_\text{eff}^\text{H}$ introduces off-diagonal components.

Unlike existing works on parameter estimation on AFDM-ISAC \cite{zhang2026CRBpilot,bao2024performance,cao2025agile,ni2025twc,luo2025novel,wang2025afdmLEO}, which treat all transmitted symbols as known deterministic quantities and thus reduce the covariance to $\sigma_\text{n}^2\mathbf{I}$, the derivation herein explicitly retains the data-induced covariance term, yielding a more generalized and accurate characterization of the achievable estimation performance for the matched models. Based on the true PDF $p_Y(\mathbf{y})$, the log-likelihood function is given by
\begin{align}
    l(\bar{\bm{\theta}}) 
    =& \ln{p_Y(\mathbf{y})} \notag \\
    =& -(\mathbf{y}-\bm{\mu}(\bar{\bm{\theta}}))^\text{H} \mathbf{C}^{-1}(\bar{\bm{\theta}}) (\mathbf{y}-\bm{\mu}(\bar{\bm{\theta}})) \notag \\
    &- \ln{ \det{\mathbf{C}(\bar{\bm{\theta}})} }- N\ln{\pi}. 
    \label{loglike_true}
\end{align}
The Fisher information matrix (FIM) is defined as
\begin{align}
    \mathbf{F}_{\bar{\bm{\theta}}} \triangleq 
    -\mathbb{E}_p \left[ \frac{\partial^2 l(\bar{\bm{\theta}}) }{\partial\bar{\bm{\theta}}\partial\bar{\bm{\theta}}^\text{T}} \right],
    \label{FIM_def}
\end{align}
where the $(m,n)$-th entry, according to the Slepian-Bangs formula \cite{slepian1954estimation,bangs1971array}, is given by
\begin{align}
    [\mathbf{F}_{\bar{\bm{\theta}}}]_{m,n} = &
    2\Re\left\{ 
    \frac{\partial \bm{\mu}(\bar{\bm{\theta}})^\text{H}}
    {\partial [\bar{\bm{\theta}}]_m} 
    \mathbf{C}^{-1}(\bar{\bm{\theta}}) 
    \frac{\partial \bm{\mu}(\bar{\bm{\theta}})}
    {\partial [\bar{\bm{\theta}}]_n} 
    \right\} \notag\\
    &+\text{tr} \left( \mathbf{C}^{-1}(\bar{\bm{\theta}}) 
    \frac{\partial \mathbf{C}(\bar{\bm{\theta}})}
    {\partial [\bar{\bm{\theta}}]_m} 
    \mathbf{C}^{-1}(\bar{\bm{\theta}}) 
    \frac{\partial \mathbf{C}(\bar{\bm{\theta}})}
    {\partial [\bar{\bm{\theta}}]_n} 
    \right). \label{FIM_res}
\end{align}
Specific expressions for the derivative terms in (\ref{FIM_res}) are provided in Appendix \ref{Appen_deriv_true}. The CRB can then be obtained as
\begin{align}
     \text{CRB}(\bar{\bm{\theta}}) \triangleq 
    \mathbf{F}_{\bar{\bm{\theta}}}^{-1}.
    \label{CRB_def}
\end{align}

\begin{remark}
    For random or unknown data, the data-induced term in (\ref{Cov_true}) becomes dependent on the parameters to be estimated. Thus, the FIM in (\ref{FIM_res}) contains both mean- and covariance-derivative terms, with the latter reflecting parameter information conveyed through second-order statistics. In contrast, when the data are known, their contribution enters the mean in (\ref{E_true}), and the covariance reduces to the parameter-independent noise term, leaving only the mean-derivative term, as in \cite[Eq.~(29)]{zhang2026CRBpilot}. Since the data realizations are more informative than their statistics alone, the CRB with known data is generally lower than its unknown-data counterpart.
\end{remark}

\subsection{MCRB Under Model Mismatch} \label{subsec:MCRB}
Under a mismatched model, the MCRB is adopted to characterize the performance bounds on parameter estimation. For a general model-mismatch case, where the assumed PDF is $f_Y(\mathbf{y}\mid\boldsymbol{\theta})=\mathcal{CN}(\bm{\mu}_\text{a}(\bm{\theta}),\mathbf{C}_\text{a}(\bm{\theta}))$ with the true PDF given by $p_Y(\mathbf{y})=\mathcal{CN}(\bm{\mu}(\bm{\bar{\theta}}),\mathbf{C}(\bm{\bar{\theta}}))$, the log-likelihood function is formulated as
\begin{align}
    l_\text{a}(\bm{\theta}) 
    &= \ln{f_Y(\mathbf{y}|\bm{\theta})} \notag \\
    =& -(\mathbf{y}-\bm{\mu}_\text{a}(\bm{\theta}))^\text{H} \mathbf{C}^{-1}_\text{a}(\bm{\theta}) (\mathbf{y}-\bm{\mu}_\text{a}(\bm{\theta})) \notag \\
    &- \ln{ \det{\mathbf{C}_\text{a}(\bm{\theta})} }- N\ln{\pi}. 
    \label{loglike_gen} 
\end{align}
By defining $\mathbf{J}_{\Tilde{\bm{\theta}}}$ and $\mathbf{K}_{\Tilde{\bm{\theta}}}$ as
\begin{align}
    \mathbf{J}_{\Tilde{\bm{\theta}}} \triangleq 
    -\mathbb{E}_p \left[ \frac{\partial^2 l_\text{a}(\bm{\theta})}{\partial\bm{\theta}\partial\bm{\theta}^\text{T}} 
    \Bigg| _{\bm{\theta}=\Tilde{\bm{\theta}}} \right]
    \label{J_def_gen}
\end{align}
and
\begin{align}
    \mathbf{K}_{\Tilde{\bm{\theta}}} \triangleq 
    \mathbb{E}_p \left[ \frac{\partial l_\text{a}(\bm{\theta})}{\partial\bm{\theta}} 
    \left( \frac{\partial l_\text{a}(\bm{\theta})}{\partial\bm{\theta}} \right)^\text{T} 
    \Bigg| _{\bm{\theta}=\Tilde{\bm{\theta}}} \right],
    \label{K_def_gen}
\end{align}
the MCRB is then given by \cite{MCRB}
\begin{align}
    \text{MCRB}(\Tilde{\bm{\theta}}) \triangleq 
    \mathbf{J}_{\Tilde{\bm{\theta}}}^{-1} \mathbf{K}_{\Tilde{\bm{\theta}}} \mathbf{J}_{\Tilde{\bm{\theta}}}^{-1}.
    \label{MCRB_def_gen}
\end{align}
Here, the MCRB is evaluated at $\Tilde{\bm{\theta}}$, which denotes the pseudotrue parameter vector assumed to be an interior point of the entire parameter space, denoted as $\bm{\Theta}=\bm{\Theta}_1\times\cdots\times\bm{\Theta}_P$. It is defined as the minimizer of the KLD between the true and assumed PDF:
\begin{align}
    \Tilde{\bm{\theta}} \triangleq& 
    \arg\min_{\bm{\theta} \in \bm{\Theta}} D_\text{KL}(p_Y||f_{Y|\bm{\theta}}) \notag\\
    =& \arg\min_{\bm{\theta} \in \bm{\Theta}} -\mathbb{E}_p[l_\text{a}(\bm{\theta})], 
    \label{pseudo_para_def_gen}
\end{align}
where 
$D_\text{KL}(p_Y||f_{Y|\bm{\theta}})$ represents the KLD between $p_Y(\mathbf{y})$ and $f_Y(\mathbf{y}|\bm{\theta})$. 
Under model mismatch, the misspecified maximum likelihood (ML) estimator no longer converges to the true parameter $\bar{\boldsymbol{\theta}}$, but instead converges to the pseudotrue parameter $\tilde{\boldsymbol{\theta}}$, which is termed as the misspecified-unbiased estimator \cite{MCRB}.
Consequently, the misspecified-unbiased estimator is generally biased with $\tilde{\boldsymbol{\theta}} - \bar{\boldsymbol{\theta}}$, and the MCRB provides an LB on the estimator variance (and the MSE if $\tilde{\boldsymbol{\theta}} - \bar{\boldsymbol{\theta}}=\mathbf{0}$) evaluated at $\tilde{\boldsymbol{\theta}}$.
Note that the expectation operator $\mathbb{E}_p[\cdot]$ throughout these expressions is taken w.r.t. the true PDF $p_Y(\mathbf{y})$. 

We next specialize this general MCRB formulation to two practical model-mismatch cases arising in AFDM-ISAC parameter estimation.

\subsubsection{Case 1: Covariance Mismatch Due to Data Leakage} \label{subsubsec:case1}
As shown in (\ref{Cov_true}), the covariance of the received signal contains the data-induced term $\mathbf{H}_\text{eff}(\bar{\boldsymbol{\theta}})\mathrm{Cov}[\mathbf{x}_\text{d}]\mathbf{H}_\text{eff}(\bar{\boldsymbol{\theta}})^H$, which is generally non-diagonal and parameter dependent.
With sufficient guard symbols, the pilot observation can be obtained by truncating the received signal to isolate it from surrounding data symbols \cite{AFDM,bemani2024integrated,li2026matched}. The covariance of the observation used for subsequent estimation then corresponds to a sub-matrix of the original covariance and reduces to a diagonal noise matrix.
In practice, however, when the guard symbols are insufficient, the pilot observation cannot be cleanly separated from the data, and the entire received signal must be retained as the observation.
Since conventional matched-filtering estimators \cite{bemani2024integrated,li2026matched,sctPilot} typically assume a scaled identity covariance, they produce a covariance mismatch.

Assuming a known number of targets and applying the general MCRB framework to this mismatch scenario, we have the assumed PDF $f_Y(\mathbf{y}|\boldsymbol{\theta}) = \mathcal{CN}(\bm{\mu}(\bm{\theta}),\mathbf{C}_\text{a})$. The assumed covariance matrix can be expressed by
\begin{align}
    \mathbf{C}_\text{a} 
    =\sigma_\text{a}^2\mathbf{I} + \sigma_{\text{n}}^2\mathbf{I}, \label{Cov_assume_case1}
\end{align}
where $\sigma_\text{a}^2$ is the assumed received data power.
Then, the log-likelihood function is given by 
\begin{align}
    l_\text{a}(\bm{\theta}) 
    =& -(\mathbf{y}-\bm{\mu}(\bm{\theta}))^\text{H} \mathbf{C}^{-1}_\text{a}(\mathbf{y}-\bm{\mu}(\bm{\theta})) \notag \\
    &- \ln{ \det{\mathbf{C}_\text{a}} }- N\ln{\pi},
    \label{loglike_case1} 
\end{align}
where only the first term is related to the parameter $\bm{\theta}$.

\begin{lemma}
Under Case 1, the pseudotrue parameter coincides with the true parameter, i.e., 
\begin{align} 
    \Tilde{\bm{\theta}} = \arg\min_{\bm{\theta} \in \bm{\Theta}} 
    [\bm{\mu}(\bar{\bm{\theta}})-\bm{\mu}(\bm{\theta})]^\text{H} [\bm{\mu}(\bar{\bm{\theta}})-\bm{\mu}(\bm{\theta})] = \bar{\bm{\theta}}. \label{pseudo_para_case1}
\end{align}
\end{lemma}
\begin{proof}
    See Appendix \ref{Appen_JK_case1}.
\end{proof}
Meanwhile, the $(m,n)$-th entries of (\ref{J_def_gen}) and (\ref{K_def_gen}) under this case are respectively given by
\begin{align}
    [\mathbf{J}_{\Tilde{\bm{\theta}}}]_{m,n} 
    =& 2\Re\left\{ 
    \frac{\partial \bm{\mu}(\bm{\theta})^\text{H}}
    {\partial [\bm{\theta}]_m} 
    \mathbf{C}_\text{a}^{-1} 
    \frac{\partial \bm{\mu}(\bm{\theta})}
    {\partial [\bm{\theta}]_n} 
    \right\} 
    \Bigg| _{\bm{\theta}=\Tilde{\bm{\theta}}} \notag \\
    \label{J_res_case1}
\end{align}
and
\begin{align}
    [\mathbf{K}_{\Tilde{\bm{\theta}}}]_{m,n} 
    = 2\Re\left\{
    \frac{\partial \bm{\mu}(\bm{\theta})^\text{H}}
    {\partial [\bm{\theta}]_m} 
    \mathbf{C}_\text{a}^{-1} 
    \mathbf{C}(\bar{\bm{\theta}})
    \mathbf{C}_\text{a}^{-1} 
    \frac{\partial \bm{\mu}(\bm{\theta})}
    {\partial [\bm{\theta}]_n} 
    \right\}
    \Bigg| _{\bm{\theta}=\Tilde{\bm{\theta}}}. \label{K_res_case1}
\end{align}

Since $\Tilde{\bm{\theta}} =\bar{\bm{\theta}}$ holds, the specific expressions of the derivative terms evaluated at $\tilde{\bm{\theta}}$ in (\ref{J_res_case1}) and (\ref{K_res_case1}) are the same as those evaluated at $\bar{\bm{\theta}}$ given in Appendix \ref{Appen_deriv_true}.

\begin{remark}
    Since the assumed covariance matrix $\mathbf{C}_\text{a}$ is a constant multiple of the identity matrix as given by (\ref{Cov_assume_case1}), its inverse $\mathbf{C}_\text{a}^{-1}$ can be factored out as a scalar from the expressions of both $\mathbf{J}_{\Tilde{\bm{\theta}}}$ and $\mathbf{K}_{\Tilde{\bm{\theta}}}$. These scalar factors perfectly cancel each other out during the computation of $\mathbf{J}_{\Tilde{\bm{\theta}}}^{-1} \mathbf{K}_{\Tilde{\bm{\theta}}} \mathbf{J}_{\Tilde{\bm{\theta}}}^{-1}$ in (\ref{MCRB_def_gen}). Consequently, the resulting MCRB is independent of $\mathbf{C}_\text{a}^{-1}$, meaning it is insensitive to the specific values of $\sigma_\text{a}^2$.
\end{remark}

\subsubsection{Case 2: Combined Covariance-and-Mean Mismatch Under Sequential Single-Target Estimation} \label{subsubsec:case2}
When the number of targets is unknown a priori, joint estimation of all target parameters is infeasible.
Instead, a sequential single-target estimation scheme is commonly adopted \cite{AFDM, li2026matched, sctPilot}. 
When estimating the $i$-th target, the contributions of the previously estimated targets have been successively but imperfectly removed from the received signal, whereas the remaining unestimated targets are absent from the assumed signal mean.
This imperfect cancellation and omission introduce a mean mismatch, while the covariance is still approximated by $\mathbf{C}_\text{a}$ as in Case 1, leading to a combined covariance-and-mean mismatch.

In this case, the assumed PDF is $f_Y(\mathbf{y}|\boldsymbol{\theta}) = \mathcal{CN}(\bm{\mu}_{\text{a}}(\bm{\theta}_{i}),\mathbf{C}_\text{a})$, where the assumed mean is formed as
\begin{align}
    \bm{\mu}_{\text{a}}(\bm{\theta}_{i})
    =& \mathbf{y}_\text{p}(\bm{\theta}_{i}) + \sum_{j=1}^{i-1}\mathbf{y}_\text{p}(\Tilde{\bm{\theta}}_{j})
    \notag\\
    =& h_{i} \mathbf{A} \mathbf{\Delta}_{k_{i}} \mathbf{\Pi}^{l_{i}} \mathbf{A}^\text{H} \mathbf{x}_\text{p} 
    + \sum_{j=1}^{i-1} \Tilde{h}_{j} \mathbf{A} \mathbf{\Delta}_{\Tilde{k}_{j}} \mathbf{\Pi}^{\Tilde{l}_{j}} \mathbf{A}^\text{H} \mathbf{x}_\text{p}, \label{E_assume_case2}
\end{align}
which deviates from the true mean $\bm{\mu}(\bm{\bar{\theta}})=\sum_{i=1}^P \mathbf{y}_\text{p}(\bm{\bar{\theta}}_i)$. 
Since the misspecified ML estimator converges to the pseudotrue parameter as the number of observations increases, the previous estimates are replaced by the corresponding pseudotrue parameters $\tilde{\boldsymbol{\theta}}_j$ in (\ref{E_assume_case2}), which are usually biased w.r.t the true parameters.

The log-likelihood function in this case is given by 
\begin{align}
    l_\text{a}(\bm{\theta}_i) 
    =& -(\mathbf{y}-\bm{\mu}_\text{a}(\bm{\theta}_i))^\text{H} \mathbf{C}^{-1}_\text{a}(\mathbf{y}-\bm{\mu}_\text{a}(\bm{\theta}_i)) \notag \\
    &- \ln{ \det{\mathbf{C}_\text{a}} }- N\ln{\pi}.
    \label{loglike_case2} 
\end{align}
\begin{lemma}
Under Case 2, the pseudotrue parameter is characterized as follows:
\begin{align}
    \Tilde{\bm{\theta}}_i 
    =& \arg\min_{\bm{\theta}_i \in \bm{\Theta}_i} -\mathbb{E}_p[l_\text{a}(\bm{\theta}_i)] \notag\\
    =& \arg\min_{\bm{\theta}_i \in \bm{\Theta}_i} 
    (\bm{\mu}(\bar{\bm{\theta}}) - \bm{\mu}_\text{a}(\bm{\theta}_i))^\text{H}(\bm{\mu}(\bar{\bm{\theta}}) - \bm{\mu}_\text{a}(\bm{\theta}_i)), \label{pseudo_para_case2}
\end{align}
\end{lemma}
\begin{proof}
    See Appendix \ref{Appen_JK_case2}.
\end{proof}
It should be noted that $\tilde{\boldsymbol{\theta}}_i$ no longer coincides with the true parameter $\bar{\boldsymbol{\theta}}_i$ of the $i$-th target because of interference from the unestimated targets and bias in the already estimated targets.

Evaluated at $\Tilde{\bm{\theta}}_i$, the MCRB has the following form:
\begin{align}
    \text{MCRB}_i(\Tilde{\bm{\theta}}_i) \triangleq 
    \mathbf{J}_{\Tilde{\bm{\theta}}_i}^{-1} \mathbf{K}_{\Tilde{\bm{\theta}}_i} \mathbf{J}_{\Tilde{\bm{\theta}}_i}^{-1}.
    \label{MCRB_def_case2} 
\end{align}
Using the derivation provided in Appendix \ref{Appen_JK_case2}, the $(m,n)$-th entries of $\mathbf{J}_{\Tilde{\bm{\theta}}_i}$ and $\mathbf{K}_{\Tilde{\bm{\theta}}_i}$ are respectively given by
\begin{align}
    [\mathbf{J}_{\Tilde{\bm{\theta}}_i}]_{m,n} 
    \triangleq &
    -\mathbb{E}_p \left[ \frac{\partial^2 l_\text{a}(\bm{\theta}_{i})}{\partial[\bm{\theta}_{i}]_m\partial[\bm{\theta}_{i}]_n} 
    \Bigg| _{\bm{\theta}_{i}=\Tilde{\bm{\theta}}_i} \right] \notag\\
    =& 2\Re\left\{ 
    \frac{\partial \bm{\mu}_{\text{a}}(\bm{\theta}_{i})^\text{H}}
    {\partial [\bm{\theta}_{i}]_m} 
    \mathbf{C}_\text{a}^{-1} 
    \frac{\partial \bm{\mu}_{\text{a}}(\bm{\theta}_{i})}
    {\partial [\bm{\theta}_{i}]_n} 
    \right\} 
    \Bigg| _{\bm{\theta}_{i}=\Tilde{\bm{\theta}}_i} \notag \\
    -& 2\Re\left\{ 
    \frac{\partial^2 \bm{\mu}_{\text{a}}(\bm{\theta}_{i})^\text{H}}
    {\partial [\bm{\theta}_{i}]_m \partial [\bm{\theta}_{i}]_n} 
    \mathbf{C}_\text{a}^{-1} 
    (\bm{\mu} (\bar{\bm{\theta}}) - \bm{\mu}_\text{a} (\bm{\theta}_i) )\right\}
    \Bigg| _{\bm{\theta}_{i}=\Tilde{\bm{\theta}}_i} \label{J_res_case2}
\end{align}
and
\begin{align}
    [\mathbf{K}_{\Tilde{\bm{\theta}}_i}]_{m,n} 
    \triangleq &
    \mathbb{E}_p \left[ \frac{\partial l_\text{a}(\bm{\theta}_{i})}{\partial[\bm{\theta}_{i}]_m} 
     \frac{\partial l_\text{a}(\bm{\theta}_{i})}{\partial[\bm{\theta}_{i}]_n}  
    \Bigg| _{\bm{\theta}_{i}=\Tilde{\bm{\theta}}_i} \right] \notag\\
    =& 2\Re\left\{
    \frac{\partial \bm{\mu}_{\text{a}}(\bm{\theta}_{i})^\text{H}}
    {\partial [\bm{\theta}_{i}]_m} 
    \mathbf{C}_\text{a}^{-1} 
    \mathbf{C}(\bar{\bm{\theta}})
    \mathbf{C}_\text{a}^{-1} 
    \frac{\partial \bm{\mu}_{\text{a}}(\bm{\theta}_{i})}
    {\partial [\bm{\theta}_{i}]_n} 
    \right\}
    \Bigg| _{\bm{\theta}_{i}=\Tilde{\bm{\theta}}_i}. \label{K_res_case2}
\end{align}
The specific expressions of the first and second derivative terms in (\ref{J_res_case2}) and (\ref{K_res_case2}) are provided in Appendix \ref{Appen_deriv_assume}.

\begin{remark}
    Like Case~1, the resulting $\text{MCRB}_i$ is independent of the scalar magnitude of $\mathbf{C}_\text{a}$. However, Case~2 yields a separate $\text{MCRB}_i$ for each target and generally introduces estimation bias. The corresponding per-target MSE bounds incorporating this bias are developed in the next section.
\end{remark}

\section{MSE and Pseudotrue Parameter Analysis} \label{sec:MSE}
This section derives an MSE lower bound by combining the MCRB and the estimation bias, and then characterizes the pseudotrue parameter and its associated bias under Case 2 for different prior-knowledge conditions.

By defining the bias vector as $\mathbf{b} \triangleq \bar{\bm{\theta}} - \Tilde{\bm{\theta}}$, the MSE of the estimate of $\bar{\bm{\theta}}$ has the following LB\cite{MCRB}:
\begin{align}
    \text{MSE} \geq \text{MCRB}(\Tilde{\bm{\theta}}) + \mathbf{b} \mathbf{b}^\text{T} \triangleq \text{LB},
    \label{MSE_LB}
\end{align}
where $\text{MSE}\triangleq\mathbb{E}_p[(\hat{\bm{\theta}}-\bar{\bm{\theta}})(\hat{\bm{\theta}}-\bar{\bm{\theta}})^\text{T}]$ and $\hat{\bm{\theta}}$ denotes the estimate of $\bar{\bm{\theta}}$.
This LB consists of two components: the MCRB, characterizing the variance floor under model mismatch, and the squared bias $\mathbf{b} \mathbf{b}^\text{T}$, which is non-zero whenever the misspecified estimator converges to a pseudotrue parameter $\tilde{\bm{\theta}} \neq \bar{\bm{\theta}}$. The behavior of these two components differs across the cases analyzed in Sec. \ref{sec:CRBandMCRB}.

For the matched model in Sec. \ref{subsec:CRB}, the assumed PDF coincides with the true PDF, which implies $\tilde{\bm{\theta}} = \bar{\bm{\theta}}$ and hence $\mathbf{b}=\mathbf{0}$.  
Therefore, the MCRB and LB in (\ref{MSE_LB}) reduces to the standard CRB when the model is correctly specified, resulting in $\text{MSE} \geq \text{CRB}(\bar{\bm{\theta}})$, with equality achieved by the ML estimator.

In Case~1 of Sec. \ref{subsec:MCRB}, although the assumed covariance deviates from the true covariance, the assumed mean remains equal to the true mean, i.e., $\bm{\mu}_\text{a}(\bm{\theta}) = \bm{\mu}(\bm{\theta})$ for all $\bm{\theta}$. Consequently, the pseudotrue parameter satisfies $\tilde{\bm{\theta}} = \bar{\bm{\theta}}$ regardless of the covariance mismatch as shown in (\ref{pseudo_para_case1}), and it can also achieve $\text{LB} = \text{MCRB}(\bar{\bm{\theta}})$ with zero bias, resulting in $\text{MSE} \geq \text{MCRB}(\bar{\bm{\theta}})$.

In Case~2, however, the additional mean mismatch introduced by sequential single-target 
estimation causes $\tilde{\bm{\theta}}_i \neq \bar{\bm{\theta}}_i$ in 
general, giving rise to a non-zero estimation bias. For the $i$-th target, it follows that
\begin{align}
    \text{MSE}_i \geq \text{MCRB}_i(\Tilde{\bm{\theta}}_i) + \mathbf{b}_i \mathbf{b}_i^\text{T} \triangleq \text{LB}_i,
    \label{MSE_LB_case2}
\end{align}
where $\text{MSE}_i\triangleq\mathbb{E}_p[(\hat{\bm{\theta}}_i-\bar{\bm{\theta}}_i)(\hat{\bm{\theta}}_i-\bar{\bm{\theta}}_i)^\text{T}]$, $\mathbf{b}_i \triangleq \bar{\bm{\theta}}_i - \Tilde{\bm{\theta}}_i$ and $\hat{\bm{\theta}}_i = [\Re(\hat{h}_i),\Im(\hat{h}_i),\hat{k}_i,\hat{l}_i]^\text{T}$ denotes the estimate of $\bar{\bm{\theta}}_i$.
To evaluate the LB of $\text{MSE}_i$, the remainder of this section focuses on characterizing the pseudotrue parameter $\tilde{\bm{\theta}}_i$ and its associated bias $\mathbf{b}_i$ under Case 2. 

For analytical convenience, we define $Q(\bm{\theta}_i) \triangleq (\bm{\mu}(\bar{\bm{\theta}}) - \bm{\mu}_\text{a}(\bm{\theta}_i))^\text{H}(\bm{\mu}(\bar{\bm{\theta}}) - \bm{\mu}_\text{a}(\bm{\theta}_i))$ from (\ref{pseudo_para_case2}). 
Under different prior knowledge of the parameters, the analytical form of $\Tilde{\bm{\theta}}_i = \arg\min_{\bm{\theta}_i \in \bm{\Theta}_i} Q(\bm{\theta}_i)$ varies as investigated below.

\subsection{Perfectly Known Doppler and Delay}
We first consider the simplest case, where the receiver has perfect knowledge of the exact Doppler shift $\bar{k}_i$ and delay $\bar{l}_i$ of the $i$-th target. 
The parameter space thus reduces to $\bm{\Theta}_i|_{k_i = \bar{k}_i, l_i = \bar{l}_i} = \bm{\Theta}_h\times\{\bar{k}_i\}\times\{\bar{l}_i\}$, where $\bm{\Theta}_h = \mathbb{R}^2$. 
The only unknown parameter is the complex gain $h_i$, reducing $\bm{\theta}_i$ to a two-dimensional real vector $[\Re\{h_i\},\Im\{h_i\}]^\text{T}$. 
Consequently, only the $2\times2$ sub-block of (\ref{MSE_LB_case2}) is relevant in this case, i.e.,
\begin{align}
    \text{MSE}_{i}^{(h)}
    \geq
    \text{MCRB}_{i}^{(h)}(\tilde{h}_i)
    + \mathbf{b}_{i}^{(h)}\mathbf{b}_{i}^{(h) \text{T}}
    \triangleq
    \text{LB}_{i}^{(h)},
\end{align}
where $\text{MSE}_{i}^{(h)}$, $\text{MCRB}_{i}^{(h)}(\tilde{h}_i)$, $\text{LB}_{i}^{(h)}$ and $\mathbf{b}_{i}^{(h)}$ denote the corresponding sub-blocks and sub-vector.

The assumed signal mean becomes $\bm{\mu}_\text{a}(\bm{\theta}_i) = h_i \bar{\mathbf{v}}_i + \sum_{j=1}^{i-1}\tilde{h}_j \bar{\mathbf{v}}_j$,  where $\bar{\mathbf{v}}_i \triangleq \mathbf{A} \mathbf{\Delta}_{\bar{k}_i} \mathbf{\Pi}^{\bar{l}_i} \mathbf{A}^\text{H} \mathbf{x}_\text{p}$.
The optimization problem in (\ref{pseudo_para_case2}) can then be simplified to finding the complex scalar $h_i$ that minimizes the following cost function:
\begin{align}
    Q(h_i) = (\bm{\mu}_i(\bar{\bm{\theta}}) - h_i \bar{\mathbf{v}}_i)^\text{H}(\bm{\mu}_i(\bar{\bm{\theta}}) - h_i \bar{\mathbf{v}}_i), \label{pseudo_fun_h}
\end{align}
where $\bm{\mu}_i(\bar{\bm{\theta}}) = \bm{\mu}(\bar{\bm{\theta}}) - \sum_{j=1}^{i-1}\tilde{h}_j \bar{\mathbf{v}}_j$.
By expanding (\ref{pseudo_fun_h}) and making the partial derivatives equal to $0$, 
the closed-form solution for the pseudotrue complex gain can be obtained as
\begin{align}
    \Tilde{h}_i = \frac{\bar{\mathbf{v}}_i^\text{H} \bm{\mu}_i(\bar{\bm{\theta}})}{||\bar{\mathbf{v}}_i||^2}. \label{pseudo_res_h_klKnw}
\end{align}
\begin{remark}
    By substituting $\bm{\mu}(\bar{\bm{\theta}}) = \bar{h}_i \bar{\mathbf{v}}_i + \sum_{j=1,j\neq i}^P \bar{h}_j \bar{\mathbf{v}}_j$ into $\bm{\mu}_i(\bar{\bm{\theta}})$ in (\ref{pseudo_res_h_klKnw}), $\Tilde{h}_i$ can be explicitly decomposed into three components:
    \begin{align}
        \Tilde{h}_i = \bar{h}_i 
        + \sum_{j= 1}^{i-1} (\bar{h}_j - \tilde{h}_j)
        \frac{\bar{\mathbf{v}}_i^\text{H} \bar{\mathbf{v}}_j}
        {||\bar{\mathbf{v}}_i||^2}
        + \sum_{j= i+ 1}^P \bar{h}_j 
        \frac{\bar{\mathbf{v}}_i^\text{H} \bar{\mathbf{v}}_j}
        {||\bar{\mathbf{v}}_i||^2}.
        \label{pseudo_res_h_klKnw_decom}
    \end{align}
    This result indicates that even when the delay and Doppler shift are perfectly known, the pseudotrue complex gain can remain biased under sequential single-target estimation. To isolate the source of this bias, consider an idealized scenario where all previously attained pseudotrue parameters are error-free, i.e., $\tilde{h}_j = \bar{h}_j$ for $j < i$. Under this stronger assumption, $\Tilde{h}_i$ is composed of the true complex gain $\bar{h}_i$ plus the last term in (\ref{pseudo_res_h_klKnw_decom}), which represents deterministic interference from all targets not yet estimated. Designing pilot sequences that make this interference term zero or approximately zero constitutes an interesting direction for future work.
\end{remark}

\subsection{Known Delay but Unknown Doppler}
Having established the case where both delay and Doppler shifts are perfectly known, we now relax the assumption on the Doppler knowledge. 
In some of the existing works on AFDM systems \cite{AFDM,EP,sctPilot,SP}, 
channels are typically characterized by integer delay taps and fractional Doppler shifts.
Consequently, the delay estimation tends to be significantly more accurate than Doppler estimation. Therefore, we next consider the case where the delay is perfectly estimated, but the complex gain and Doppler shift are still required to be estimated. 

In this case, the parameter space reduces to $\bm{\Theta}_i|_{l_i = \bar{l}_i} = \bm{\Theta}_h\times\bm{\Theta}_k\times\{\bar{l}_i\}$, where $\bm{\Theta}_h = \mathbb{R}^2$ and $\bm{\Theta}_k = [-k_\text{max}, k_\text{max}]$. 
Accordingly, $\bm{\theta}_i$ reduces to a three-dimensional real vector $[\Re\{h_i\},\Im\{h_i\},k_i]^\text{T}$ and the $3\times3$ sub-block of (\ref{MSE_LB_case2}) is relevant, given by
\begin{align}
    \text{MSE}_i^{(h,k)}
    \geq
    \text{MCRB}_i^{(h,k)}(\tilde{h}_i,\tilde{k}_i)
    + \mathbf{b}_i^{(h,k)}\mathbf{b}_i^{(h,k)\text{T}}
    \triangleq
    \text{LB}_i^{(h,k)},
\end{align}
where $\text{MSE}_i^{(h,k)}$, $\text{MCRB}_i^{(h,k)}(\tilde{h}_i,\tilde{k}_i)$, $\text{LB}_i^{(h,k)}$ and $\mathbf{b}_i^{(h,k)}$ denote the corresponding sub-blocks and sub-vector.

The assumed mean can be expressed as $\bm{\mu}_\text{a}(\bm{\theta}_i) = h_i \mathbf{v}(k_i) + \sum_{j=1}^{i-1}\tilde{h}_j \mathbf{v}(\tilde{k}_j)$,  where $\mathbf{v}(k_i)\triangleq \mathbf{A} \mathbf{\Delta}_{k_i} \mathbf{\Pi}^{\bar{l}_i} \mathbf{A}^\text{H} \mathbf{x}_\text{p}$. 
The optimization problem in (\ref{pseudo_para_case2}) can then be simplified to finding the complex scalar $h_i$ and real scalar $k_i$ that minimizes the cost function as follows:
\begin{align}
    Q(h_i,k_i) = (\bm{\mu}_i(\bar{\bm{\theta}}) - h_i \mathbf{v}(k_i))^\text{H}(\bm{\mu}_i(\bar{\bm{\theta}}) - h_i \mathbf{v}(k_i)), \label{pseudo_fun_hk}
\end{align}
where $\bm{\mu}_i(\bar{\bm{\theta}}) = \bm{\mu}(\bar{\bm{\theta}}) - \sum_{j=1}^{i-1}\tilde{h}_j \mathbf{v}(\tilde{k}_j)$.

To solve (\ref{pseudo_para_case2}) with the cost function in (\ref{pseudo_fun_hk}), we apply the separable least squares approach by first optimizing the linear parameter $h_i$. By taking the partial derivatives of $Q(h_i, k_i)$ with respect to $\Re\{h_i\}$ and $\Im\{h_i\}$ and equating them to zero, the conditional optimal complex gain for a given Doppler shift $k_i$ can be expressed as:
\begin{align}
    \Tilde{h}_i(k_i) = \frac{\mathbf{v}(k_i)^\text{H} \bm{\mu}_i(\bar{\bm{\theta}})}
    {||\mathbf{v}(k_i)||^2}. \label{pseudo_res_hk_lKnw}
\end{align}

By substituting $\Tilde{h}_i(k_i)$ back into (\ref{pseudo_fun_hk}), the complex gain $h_i$ can be eliminated, and the cost function solely relying on the Doppler shift $k_i$ is obtained as:
\begin{align}
    Q(k_i) = ||\bm{\mu}_i(\bar{\bm{\theta}})||^2 - 
    \frac{|\mathbf{v}(k_i)^\text{H}\bm{\mu}_i(\bar{\bm{\theta}})|^2}{||\mathbf{v}(k_i)||^2}. \label{pseudo_fun_hk_onlyk}
\end{align}
Since the term $\|\bm{\mu}(\bar{\boldsymbol{\theta}})\|^2$ is independent of $k_i$, minimizing $Q(k_i)$ in (\ref{pseudo_fun_hk_onlyk}) is mathematically equivalent to maximizing its second term.
Furthermore, in AFDM systems, the matrix $\mathbf{A}$ is a unitary matrix. Given that both the Doppler phase shift matrix $\boldsymbol{\Delta}_{k_i}$ and the delay cyclic shift matrix $\boldsymbol{\Pi}^{\bar{l}_i}$ are also unitary, the norm of the denominator in this second term simplifies to a constant, i.e., $\|\mathbf{v}(k_i)\|^2 = \|\mathbf{x}_{\text{p}}\|^2$.
Consequently, the optimal Doppler shift can be obtained by finding the peak of a matched filter output, i.e.,
\begin{align}
    \Tilde{k}_i = \arg\max_{k_i\in\bm{\Theta}_k} |\mathbf{v}(k_i)^\text{H}\bm{\mu}_i(\bar{\bm{\theta}})|^2. \label{pseudo_res_k_lKnw}
\end{align}

Once the pseudotrue Doppler shift $\Tilde{k}_i$ is determined, the corresponding pseudotrue complex gain is readily evaluated by back-substitution, yielding $\Tilde{h}_i = \Tilde{h}_i(\Tilde{k}_i)$.

\begin{remark} \label{rmk_k_expan}
    Expanding (\ref{pseudo_res_k_lKnw}) and assuming the perfect previously attained pseudotrue parameters, i.e., $(\tilde{h}_j, \tilde{k}_j)= (\bar{h}_j,\bar{k}_j)$ for $j < i$, we have 
    \begin{align}
        \Tilde{k}_i =& \arg\max_{k_i\in\bm{\Theta}_k} |\bar{h}_i \gamma(k_i,\bar{k}_i,\bar{l}_i,\bar{l}_i) \notag\\
        &+ \sum_{j=i+1}^P \bar{h}_j \gamma(k_i,\bar{k}_j,\bar{l}_i,\bar{l}_j)|^2,
        \label{pseudo_res_k_lKnw_expan}
    \end{align}
    where $\gamma(k_i,\bar{k}_i,l_i,\bar{l}_i) \triangleq \mathbf{x}_\text{p}^\text{H} \mathbf{A} \mathbf{\Pi}^{-l_i} \mathbf{\Delta}_{\bar{k}_i - k_i} \mathbf{\Pi}^{\bar{l}_i} \mathbf{A}^\text{H} \mathbf{x}_\text{p}$.
    It can be observed from (\ref{pseudo_res_k_lKnw_expan}) that when the summation term (representing the inter-target interference) equals or approaches zero, the equality $\tilde{k}_i = \bar{k}_i$ is mathematically guaranteed.
    Note that the summation term in (\ref{pseudo_res_k_lKnw_expan}) can be strictly or approximately eliminated by carefully designing dedicated sequences with ideal correlation properties. The design of pilot sequences with such ideal correlation properties constitutes an interesting direction for future work.
\end{remark}

\subsection{Unknown Doppler and Delay}
Finally, we consider the most severe, yet realistic, mismatch scenario of Case 2, where neither the delay nor the Doppler shift of the target is perfectly known to the receiver before the estimation.

In this case, $\bm{\theta}_i$ satisfies $\bm{\theta}_i \in \bm{\Theta}_i  = \bm{\Theta}_h\times\bm{\Theta}_k\times\bm{\Theta}_{l}$, where $\bm{\Theta}_{l} = [0, l_\text{max}]$ denotes the feasible search space for delays. 
The assumed mean has the following form $\bm{\mu}_\text{a}(\bm{\theta}_i) = h_i \mathbf{v}(k_i,l_i) + \sum_{j=1}^{i-1}\tilde{h}_j \mathbf{v}(\tilde{k}_j, \tilde{l}_i)$,  where $\mathbf{v}(k_i,l_i) \triangleq \mathbf{A} \mathbf{\Delta}_{k_i} \mathbf{\Pi}^{l_i} \mathbf{A}^\text{H} \mathbf{x}_\text{p}$.

Following a similar separable least squares methodology as before, we can decouple the estimation of the linear complex gain from the non-linear delay and Doppler parameters. For a given pair of $k_i$ and $l_i$, the conditional optimal estimate of $h_i$ is derived by minimizing the corresponding cost function $Q(\bm{\theta}_i)$, which yields:
\begin{align}
    \Tilde{h}_i(k_i,l_i) = \frac{\mathbf{v}(k_i,l_i)^\text{H} \bm{\mu}_i(\bar{\bm{\theta}})}
    {||\mathbf{v}(k_i,l_i)||^2}, \label{pseudo_res_hkl}
\end{align}
with $\bm{\mu}_i(\bar{\bm{\theta}}) = \bm{\mu}(\bar{\bm{\theta}}) - \sum_{j=1}^{i-1}\tilde{h}_j \mathbf{v}(\tilde{k}_j,\tilde{l}_j)$.

By substituting $\Tilde{h}_i(k_i,l_i)$ back into the original optimization problem, the linear parameter $h_i$ is eliminated. The joint estimation of the Doppler shift and delay is then transformed into minimizing a similar objective function as (\ref{pseudo_fun_hk_onlyk}) with $\mathbf{v}(k_i)$ replaced by $\mathbf{v}(k_i,l_i)$. Consequently, the optimal parameter pair $(\Tilde{k}_i, \Tilde{l}_i)$ can be found via a two-dimensional joint search over the parameter space $\bm{\Theta}_k \times \bm{\Theta}_{l}$, i.e., 
\begin{align}
    (\Tilde{k}_i,\Tilde{l}_i) = \arg\max_{k_i\in\bm{\Theta}_k, l_i\in\bm{\Theta}_{l}} |\mathbf{v}(k_i,l_i)^\text{H}\bm{\mu}_i(\bar{\bm{\theta}})|^2. \label{pseudo_res_kl}
\end{align}
Subsequently, the pseudotrue complex gain is obtained by back-substituting the estimated pair $(\Tilde{k}_i,\Tilde{l}_i)$, resulting in $\Tilde{h}_i = \Tilde{h}_i(\Tilde{k}_i,\Tilde{l}_i)$.

\begin{remark}
    Similar to Remark \ref{rmk_k_expan}, (\ref{pseudo_res_kl}) can be expanded as
    \begin{align}
        (\Tilde{k}_i,\Tilde{l}_i) =& \arg\max_{k_i\in\bm{\Theta}_k, l_i\in\bm{\Theta}_{l}} 
        |\bar{h}_i \gamma(k_i,\bar{k}_i,l_i,\bar{l}_i) \notag\\
        &+ \sum_{j=i+1}^P \bar{h}_j \gamma(k_i,\bar{k}_j,l_i,\bar{l}_j)|^2
        \label{pseudo_res_kl_expan}
    \end{align}
    under the assumption of perfect previous estimates, i.e., $(\tilde{h}_j, \tilde{k}_j, \tilde{l}_j) = (\bar{h}_j, \bar{k}_j, \bar{l}_j)$ for $j < i$.
    When the summation term in (\ref{pseudo_res_kl_expan}) equals or approaches zero, it can be mathematically guaranteed that $(\tilde{k}_i, \tilde{l}_i) = (\bar{k}_i, \bar{l}_i)$. 
    The dedicated design of sequences with ideal two-dimensional correlation properties tailored to (\ref{pseudo_res_kl_expan}) remains an open topic for future work.
\end{remark}

\section{Simulation Results and Discussion} \label{sec:simu}
In this section, numerical simulations are conducted to validate the CRB and MCRB derived in Sec.~\ref{sec:CRBandMCRB} and the corresponding LB on MSE derived in Sec.~\ref{sec:MSE}. Based on this validation, the CRB, MCRB, and LB are further compared under different model-mismatch cases and system parameters. Without loss of generality, the simulation parameters are summarized in Table~\ref{tab:sim_params}. The considered AFDM system consists of $N$ subcarriers with chirp parameters $c_1 = \frac{2(k_\text{max}+k_\nu)+1}{2N}$ and $c_2 = \pi$. Unless otherwise specified, the number of targets is set to $P=3$. For each target, the delay and Doppler shift are drawn uniformly from $[0, l_\text{max}]$ and $[-k_\text{max}, k_\text{max}]$, respectively. The complex gain is randomly generated
from the complex Gaussian distribution $\mathcal{CN}(0, \frac{1}{P})$. Note that the simulation results and discussions are also applicable to channel estimation in AFDM systems.

\begin{table}[!t]
  \caption{Simulation Parameters}
  \label{tab:sim_params}
  \centering
  \begin{tabular}{lc}
    \toprule
    \textbf{Parameter} & \textbf{Value} \\
    \midrule
    Number of subcarriers $N$ & 256 \\
    Carrier frequency $f_\text{c}$ & \SI{50}{GHz} \\
    Subcarrier spacing $\Delta f$ & \SI{60}{kHz} \\
    Sampling time $T_\text{s}$ & \SI{65.1}{ns} \\
    Maximum velocity $v_\text{max}$ & \SI{648}{km/h} \\
    Maximum range $R_\text{max}$ & \SI{97.7}{m} \\
    Maximum normalized Doppler $k_\text{max}$ & 1 \\
    Maximum normalized delay $l_\text{max}$ & 10 \\
    $k_\nu$ & 4\\
    \bottomrule
  \end{tabular}
\end{table}

Pilots are placed among the $N$ subcarriers of one AFDM symbol as uniformly as possible, with no guard intervals inserted. For a given pilot length $L_\text{p} \in \{1, \dots, \lfloor N/4 \rfloor\}$, the consecutive pilot indices differ by  $d=\left\lfloor (N - L_\text{p})/L_\text{p} \right\rfloor + 1$. Any remaining subcarriers are assigned to the final segment following the last pilot, whose length may therefore exceed $d-1$. Since different values of $L_\text{p}$ may yield the same $d$, only the largest $L_\text{p}$ corresponding to each unique $d$ is retained for simulations.
Without further clarification, the power of each pilot symbol is set as \SI{30}{dB} higher than the average power of data symbols, while the data SNR is set as \SI{20}{dB}. The pilot sequence is chosen as the Zadoff-Chu (ZC) sequence with the root of $L_\text{p}-1$. 

\subsection{Validation of CRB} \label{subsec:vali_CRB}

\begin{figure}
    \centering
    \includegraphics[width=0.75\linewidth]{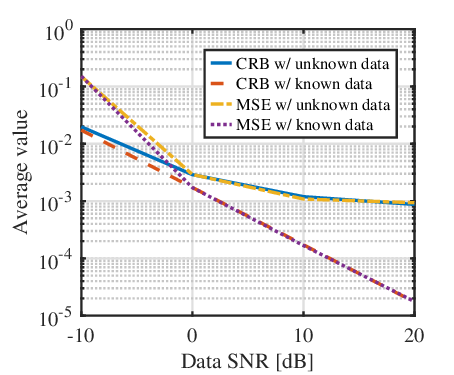}
    \caption{Validation of CRB when the data are unknown or known at the receiver.}
    \label{fig:CRB_MSE_DataSNR}
\end{figure}

We first validate the CRB derived in Sec. \ref{subsec:CRB}, by comparing it with the empirical MSE in Fig.~\ref{fig:CRB_MSE_DataSNR}. The CRB with unknown data is obtained from the FIM in (\ref{FIM_res}), whereas the CRB with known data retains only the mean-derivative term, as in \cite{zhang2026CRBpilot}.

For the empirical MSE with unknown data, the parameter vector $\bm{\theta}$ is estimated at different data SNRs by maximizing the log-likelihood in (\ref{loglike_true}) through a multidimensional search over $\bm{\theta}$. In the case of known data, a matched-filtering-based estimator as in \cite{bemani2024integrated} is adopted. Owing to the high cost of the grid searches, this validation considers a single target and uses 10 channel parameter sets and 200 independent data-and-noise realizations for each channel sample and SNR. The overall CRB and MSE shown in Fig.~\ref{fig:CRB_MSE_DataSNR} are the sums of the corresponding metrics for the four entries of $\bm{\theta}$ and are averaged over channel realizations.

Except in the low-SNR threshold region, the MSE of each matched estimator closely approaches its corresponding CRB, validating both CRB calculations. The known data provide additional Fisher information and therefore yield a lower CRB than the unknown-data case. Moreover, the CRB and MSE with known data continue to decrease with SNR, whereas the data-induced uncertainty causes those with unknown data to approach a clear high-SNR floor.

\subsection{Validation of MCRB and LB} \label{subsec:vali_MCRBLB}

\begin{figure*}[tb]
    \centering
    \begin{subfigure}{0.30\textwidth}
        \centering
        \includegraphics[width=\linewidth]{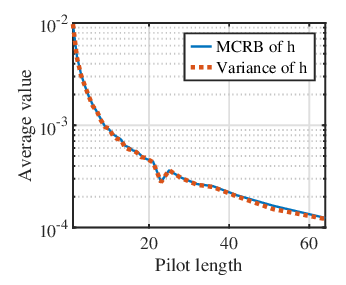}
        \caption{MCRB and variance of complex gains.}
        \label{fig:MCRB_Var_h}
    \end{subfigure}
    \hfill
    \begin{subfigure}{0.30\textwidth}
        \centering
        \includegraphics[width=\linewidth]{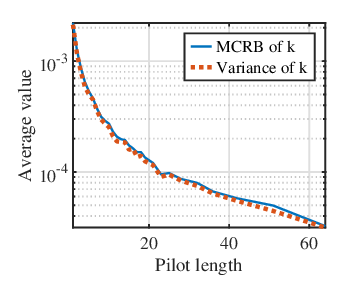}
        \caption{MCRB and variance of Doppler shifts.}
        \label{fig:MCRB_Var_k}
    \end{subfigure}
    \hfill
    \begin{subfigure}{0.30\textwidth}
        \centering
        \includegraphics[width=\linewidth]{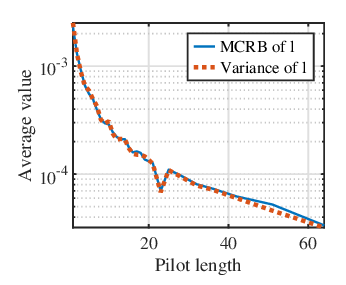}
        \caption{MCRB and variance of delays.}
        \label{fig:MCRB_Var_l}
    \end{subfigure}

    \vspace{0.5em}

    \begin{subfigure}{0.30\textwidth}
        \centering
        \includegraphics[width=\linewidth]{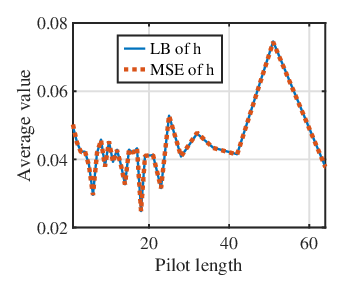}
        \caption{LB and MSE of complex gains.}
        \label{fig:LB_MSE_h}
    \end{subfigure}
    \hfill
    \begin{subfigure}{0.30\textwidth}
        \centering
        \includegraphics[width=\linewidth]{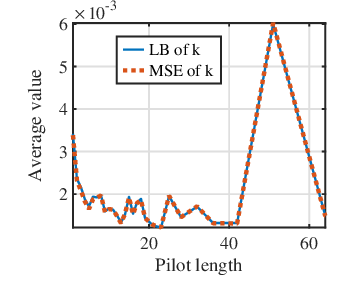}
        \caption{LB and MSE of Doppler shifts.}
        \label{fig:LB_MSE_k}
    \end{subfigure}
    \hfill
    \begin{subfigure}{0.30\textwidth}
        \centering
        \includegraphics[width=\linewidth]{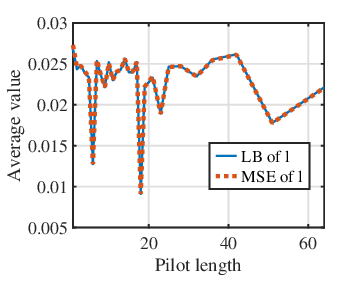}
        \caption{LB and MSE of delays.}
        \label{fig:LB_MSE_l}
    \end{subfigure}
    
    \caption{Validation of the MCRB and LB when all parameters are unknown. 
    The top row compares the MCRB with the variance of the ML estimator; the bottom row
    compares the LB with the MSE of the ML estimator for each parameter.}
    
    \label{fig:LB_MSE_MCRB_Var_compare}
\end{figure*}

Next, the MCRB and LB are compared with the empirical variance and MSE of the misspecified-unbiased, matched-filtering-based ML estimator under sequential single-target estimation. For the $i$-th target, the estimator minimizes
$(\mathbf{y}_i - h_i \mathbf{v}(k_i,l_i))^\text{H}
 (\mathbf{y}_i - h_i \mathbf{v}(k_i,l_i))$,
where
$\mathbf{y}_i=\mathbf{y} - \sum_{j=1}^{i-1}\mathbf{y}_\text{p}(\hat{\bm{\theta}}_{j})$
is the residual signal after subtracting the previously estimated targets. The delay and Doppler estimates are obtained by the two-dimensional grid search
$(\hat{k}_i,\hat{l}_i) = \arg\max_{k_i\in\bm{\Theta}_k,\, l_i\in\bm{\Theta}_{l}}
 |\mathbf{v}(k_i,l_i)^\text{H}\mathbf{y}_i|^2$,
followed by the closed-form complex-gain estimate
$\hat{h}_i(\hat{k}_i,\hat{l}_i)
 = \frac{\mathbf{v}(\hat{k}_i,\hat{l}_i)^\text{H} \mathbf{y}_i}
        {\|\mathbf{v}(\hat{k}_i,\hat{l}_i)\|^2}$.
All targets are assumed to be identified successfully, so that the comparison focuses on estimation accuracy rather than target detection.

The matrix-form LB and MCRB are reduced to per-parameter values by extracting the sub-block corresponding to the complex gain and the diagonal entries corresponding to the Doppler shift and delay in each per-target matrix. For example, the LB values are computed as $\text{LB}^{(h)} = \sum_{i=1}^{P}\mathrm{tr}\!\left([\text{LB}_i]_{1:2,\,1:2}\right)$, $\text{LB}^{(k)} = \sum_{i=1}^{P}[\text{LB}_i]_{3,3}$, and $\text{LB}^{(l)} = \sum_{i=1}^{P}[\text{LB}_i]_{4,4}$. The MCRB is processed analogously. The resulting LB, MCRB, empirical variance, and empirical MSE are then averaged over channel realizations. Owing to the high cost of the two-dimensional grid search, this validation uses 10 channel parameter sets and 200 independent data-and-noise realizations for each channel sample.

When $\{h_i, k_i, l_i\}$ are all unknown, the empirical variances closely match the MCRBs in Fig.~\ref{fig:MCRB_Var_h} to \ref{fig:MCRB_Var_l}, and the empirical MSEs closely match the corresponding LBs in Fig.~\ref{fig:LB_MSE_h} to \ref{fig:LB_MSE_l}. This consistency confirms the effectiveness of the derived bounds. The MCRB/variance curves decrease clearly as the pilot length increases, whereas the LB/MSE curves fluctuate more noticeably because they are dominated by the channel-dependent squared bias. Since this subsection only verifies the bounds, a relatively small number of channel samples is sufficient; Sec.~\ref{subsec:MCRBLBvsPilot} uses a much larger number of channel realizations to characterize the average performance trends.

\subsection{Impact of the Pilots on CRB, MCRB, and LB} \label{subsec:MCRBLBvsPilot}

\begin{figure*}[tb]
    \centering
    \begin{subfigure}{0.30\textwidth}
        \centering
        \includegraphics[width=\linewidth]{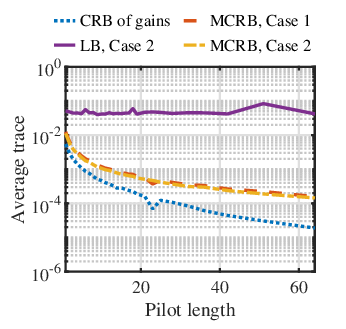}
        \caption{CRB, MCRB and LB of complex gains with all channel samples.}
        \label{fig:CRB_h_allCh}
    \end{subfigure}
    \hfill
    \begin{subfigure}{0.30\textwidth}
        \centering
        \includegraphics[width=\linewidth]{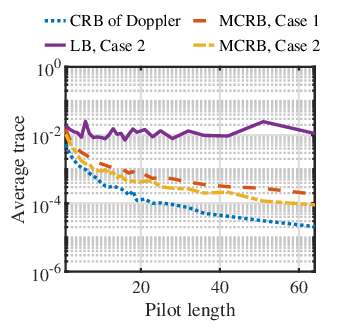}
        \caption{CRB, MCRB and LB of Doppler shifts with all channel samples.}
        \label{fig:CRB_k_allCh}
    \end{subfigure}
    \hfill
    \begin{subfigure}{0.30\textwidth}
        \centering
        \includegraphics[width=\linewidth]{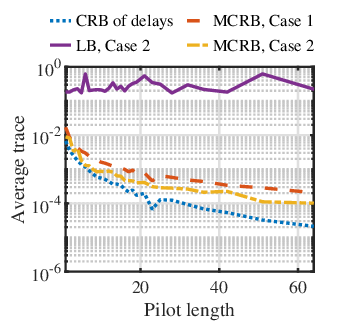}
        \caption{CRB, MCRB and LB of delays with all channel samples.}
        \label{fig:CRB_l_allCh}
    \end{subfigure}

    \vspace{0.5em}

    \begin{subfigure}{0.30\textwidth}
        \centering
        \includegraphics[width=\linewidth]{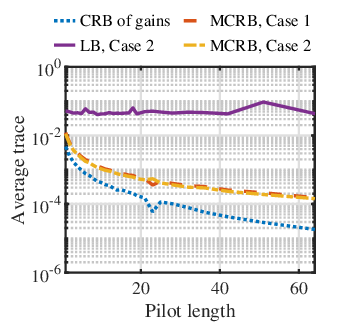}
        \caption{CRB, MCRB and LB of complex gains with weak-path samples excluded.}
        \label{fig:CRB_h_exclu}
    \end{subfigure}
    \hfill
    \begin{subfigure}{0.30\textwidth}
        \centering
        \includegraphics[width=\linewidth]{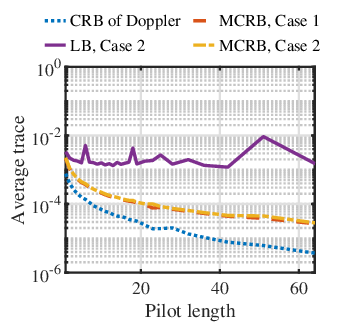}
        \caption{CRB, MCRB and LB of Doppler shifts with weak-path samples excluded.}
        \label{fig:CRB_k_exclu}
    \end{subfigure}
    \hfill
    \begin{subfigure}{0.30\textwidth}
        \centering
        \includegraphics[width=\linewidth]{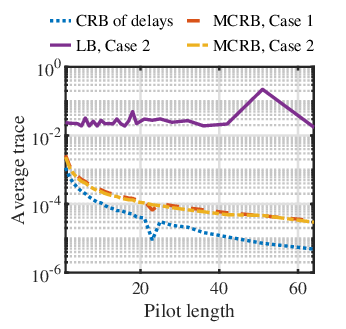}
        \caption{CRB, MCRB and LB of delays with weak-path samples excluded.}
        \label{fig:CRB_l_exclu}
    \end{subfigure}
    
    \caption{Average trace of the CRB, MCRB, and LB versus pilot length under the matched case and two mismatched cases.
    The top row shows results from all channel samples; the bottom row excludes the weak-path samples (in which any path has power below $10\%$ of the strongest path).}
    
    \label{fig:CRB_MCRB_LB_PilotLen}
\end{figure*}

Finally, we investigate the impact of pilot configurations under the matched model and the two misspecified models. Case~1 corresponds to covariance mismatch due to data leakage, for which the squared-bias term vanishes and the LB reduces to the MCRB. Case~2 additionally includes the mean mismatch caused by sequential single-target estimation, so that the LB contains a nonzero bias term. In Fig.~\ref{fig:CRB_MCRB_LB_PilotLen}, each point is averaged over $1000$ channel realizations. The first row uses all generated channel samples, whereas the second row excludes the samples in which at least one path has power below $10\%$ of the strongest path.

The main observations are summarized as follows.
\begin{itemize}
    \item The CRB gives the lowest curves under the perfectly specified model, but it is overly optimistic and no longer a valid performance bound under model mismatch. The MCRBs and especially the Case~2 LB quantify the corresponding performance loss.
    \item Increasing the pilot length reduces the CRB and MCRB curves, but it does not produce a clear downward trend for the Case~2 LB. Thus, in the sequential single-target mismatched estimator, the total error is mainly limited by bias rather than by variance.
    \item Removing weak-path samples has little effect on the complex-gain curves, but it reduces the Doppler and delay bounds by nearly one order of magnitude. Hence, weak paths dominate the Doppler/delay variance and also contribute strongly to the Case~2 bias.
    \item For Doppler and delay, after weak-path samples are excluded, the Case~1 MCRB changes from being slightly higher than the Case~2 MCRB to being close to it. This indicates that Case~1 is more sensitive to weak paths.
\end{itemize}

Some pronounced peaks of the Case~2 LB persist after the weak-path exclusion, especially around pilot length $51$. They arise from pilot-dependent sidelobes in the autocorrelation of $\mathbf{v}(k_i,l_i)$ used by the two-dimensional matched-filter search, which can approach the main peak and increase the probability of selecting incorrect grid points. Hence, pilot lengths in a smooth LB region, such as $[1,5]$ in the present configuration, yield more stable estimation performance and are preferable. Since the LB does not decrease with pilot length, the shortest stable pilot length is an efficient choice.

\begin{figure}
    \centering
    \includegraphics[width=0.75\linewidth]{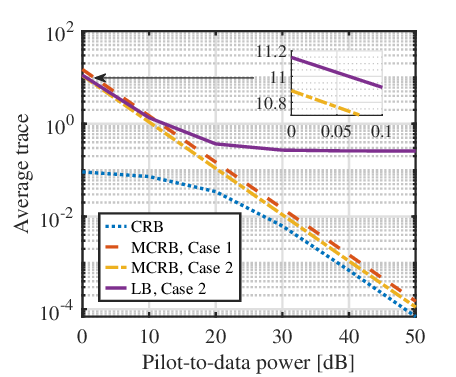}
    \caption{Average trace of the CRB, MCRB, and LB versus pilot-to-data power with pilot length of $3$.}
    \label{fig:CRB_MCRB_LB_PilotPwr}
\end{figure}

For a fixed pilot length of $3$, the average traces versus pilot-to-data power are shown in Fig.~\ref{fig:CRB_MCRB_LB_PilotPwr}. Since the individual parameters exhibit broadly similar trends in Fig.~\ref{fig:CRB_MCRB_LB_PilotLen}, the metrics are aggregated by taking the trace of the whole matrix, rather than being shown separately for each parameter. Each point is averaged over $1000$ channel realizations. The CRB and MCRB decrease almost linearly on the logarithmic scale as pilot power increases. In Case~2, however, the LB initially follows the MCRB at low pilot power, where the MCRB term dominates. As pilot power increases, the MCRB term continues to decrease while the squared-bias term remains nearly unchanged, so the LB gradually approaches a bias-dominated floor. Beyond $\SI{20}{dB}$, the LB improvement becomes marginal in the present configuration.

\section{Conclusion} \label{sec:conclusion}
In this paper, we analyzed parameter estimation in AFDM-ISAC systems under model misspecification. We extended the CRB to a general observation model with unknown data symbols and derived the MCRB for covariance mismatch and combined covariance-and-mean mismatch. We further characterized the associated pseudotrue parameters, estimation bias, and LB on the MSE. 
Beyond validating these bounds, the simulation results show that the conventional CRB is overly optimistic under model mismatch, while the MCRB and LB provide more informative variance and MSE benchmarks for practical system design.
In particular, the sequential single-target estimation used in practice introduces a comparatively large bias that does not decrease noticeably with pilot length, causing the MSE to be dominated by bias rather than variance. 
As a result, increasing pilot length or power lowers the CRB and MCRB but brings little gain once the LB reaches its bias floor, suggesting a short pilot is an efficient choice. 
Future work may investigate bias-aware pilot designs that reduce the interference from the unestimated targets, as well as joint multi-target estimation and bias-compensation methods that alleviate the mismatch introduced by sequential estimation. 
Although developed for AFDM-ISAC systems, the MCRB framework in this paper can also be extended to other waveforms by adapting the corresponding signal representation.

\appendices
\section{Derivatives of the True Mean and Covariance} \label{Appen_deriv_true}
To differentiate with respect to the delay parameter, we use the continuous frequency-domain representation
\begin{align}
    \mathbf{\Pi}^{l_i} = \mathbf{F}^\text{H} \mathbf{\Delta}_{l_i} \mathbf{F},
\end{align}
where $\mathbf{\Delta}_{l_i}=\text{diag}(e^{-j2 \pi l_i n/N},n=0,1,...,N-1)$. The resulting derivative is evaluated at integer values of $l_i$, consistent with the discrete-time models in (\ref{rs_TD}) and (\ref{yx_DAFT}).

The derivatives of $\bm{\mu}(\bar{\bm{\theta}})$ are given by
\begin{align}
    \frac{\partial \bm{\mu}(\bar{\bm{\theta}})}{\partial \Re(\bar{h}_i)} 
    &= \bar{\mathbf{B}}_i \mathbf{x}_\text{p}, &
    \frac{\partial \bm{\mu}(\bar{\bm{\theta}})}{\partial \Im(\bar{h}_i)}
    &= j \bar{\mathbf{B}}_i \mathbf{x}_\text{p},\\
    \frac{\partial \bm{\mu}(\bar{\bm{\theta}})}{\partial \bar{k}_i}
    &= \bar{h}_i \dot{\bar{\mathbf{B}}}_i^{(k)} \mathbf{x}_\text{p}, &
    \frac{\partial \bm{\mu}(\bar{\bm{\theta}})}{\partial \bar{l}_i}
    &= \bar{h}_i \dot{\bar{\mathbf{B}}}_i^{(l)} \mathbf{x}_\text{p},
\end{align}
where 
$\bar{\mathbf{B}}_i\triangleq\mathbf{A} \bm{\Delta}_{\bar{k}_i} \bm{\Pi}^{\bar{l}_i} \mathbf{A}^\text{H}$, 
$\dot{\bar{\mathbf{B}}}_i^{(k)} \triangleq -j \frac{2\pi}{N} \mathbf{A} \mathbf{D}\bm{\Delta}_{\bar{k}_i} \bm{\Pi}^{\bar{l}_i} \mathbf{A}^\text{H}$, 
$\dot{\bar{\mathbf{B}}}_i^{(l)} \triangleq -j \frac{2\pi}{N} \mathbf{A} \bm{\Delta}_{\bar{k}_i} \mathbf{F}^\text{H} \mathbf{D} \mathbf{\Delta}_{\bar{l}_i} \mathbf{F} \mathbf{A}^\text{H}$, 
$\mathbf{D}\triangleq\text{diag}(n,n=0,1,...,N-1)$ and $i=1,...,P$.

For brevity, the covariance of the transmitted data is denoted as $\mathbf{R}_x = \text{Cov}[\mathbf{x}_\text{d}]$. Then, the derivatives of $\mathbf{C}(\bar{\bm{\theta}})$ can be expressed as
\begin{align}
    \frac{\partial \mathbf{C}(\bar{\bm{\theta}})}{\partial \Re(\bar{h}_i)} 
    &= \bar{\mathbf{B}}_i \mathbf{R}_x \bar{\mathbf{H}}_\text{eff}^\text{H}
    + \bar{\mathbf{H}}_\text{eff} \mathbf{R}_x \bar{\mathbf{B}}_i^\text{H},\\
    \frac{\partial \mathbf{C}(\bar{\bm{\theta}})}{\partial \Im(\bar{h}_i)}
    &= j \bar{\mathbf{B}}_i \mathbf{R}_x \bar{\mathbf{H}}_\text{eff}^\text{H}
    - j \bar{\mathbf{H}}_\text{eff} \mathbf{R}_x \bar{\mathbf{B}}_i^\text{H},\\
    \frac{\partial \mathbf{C}(\bar{\bm{\theta}})}{\partial \bar{k}_i}
    &= \bar{h}_i \dot{\bar{\mathbf{B}}}_i^{(k)} \mathbf{R}_x \bar{\mathbf{H}}_\text{eff}^\text{H}
    + \bar{\mathbf{H}}_\text{eff} \mathbf{R}_x \bar{h}_i^*(\dot{\bar{\mathbf{B}}}_i^{(k)})^\text{H},\\
    \frac{\partial \mathbf{C}(\bar{\bm{\theta}})}{\partial \bar{l}_i}
    &= \bar{h}_i \dot{\bar{\mathbf{B}}}_i^{(l)} \mathbf{R}_x \bar{\mathbf{H}}_\text{eff}^\text{H}
    + \bar{\mathbf{H}}_\text{eff} \mathbf{R}_x \bar{h}_i^*(\dot{\bar{\mathbf{B}}}_i^{(l)})^\text{H}.
\end{align}

\section{Derivation of $\Tilde{\bm{\theta}}$, $\mathbf{J}_{\Tilde{\bm{\theta}}}$ and $\mathbf{K}_{\Tilde{\bm{\theta}}}$} \label{Appen_JK}
Based on the definitions in (\ref{pseudo_para_def_gen}), (\ref{J_def_gen}), and (\ref{K_def_gen}), the derivations of $\Tilde{\bm{\theta}}$, $\mathbf{J}_{\Tilde{\bm{\theta}}}$, and $\mathbf{K}_{\Tilde{\bm{\theta}}}$ under Case 1 and Case 2 in Sec. \ref{subsec:MCRB} are respectively provided below.

\subsubsection{Case 1} \label{Appen_JK_case1}
We begin by deriving the pseudotrue parameter $\Tilde{\bm{\theta}}$.
By substituting (\ref{loglike_case1}) into the objective function in (\ref{pseudo_para_def_gen}) and discarding the term unrelated to $\bm{\theta}$, we have
\begin{align}
    \arg\min_{\bm{\theta} \in \bm{\Theta}} -\mathbb{E}_p[l_\text{a}(\bm{\theta})] 
    = \arg\min_{\bm{\theta} \in \bm{\Theta}} 
    \mathbb{E}_p[\bm{\epsilon}(\bm{\theta})^\text{H} \mathbf{C}_\text{a}^{-1} \bm{\epsilon}(\bm{\theta})] ,
\end{align}
where $\bm{\epsilon}(\bm{\theta}) \triangleq \mathbf{y} - \bm{\mu}(\bm{\theta}) = \Tilde{\bm{\epsilon}}+\Delta\bm{\mu}(\bm{\theta})$, $\Tilde{\bm{\epsilon}} \triangleq \mathbf{y} - \bm{\mu}(\bar{\bm{\theta}})$, and $\Delta\bm{\mu}(\bm{\theta}) \triangleq \bm{\mu}(\bar{\bm{\theta}}) - \bm{\mu}(\bm{\theta})$.
Using $\mathbb{E}_p[\mathbf{y}]=\bm{\mu}(\bar{\bm{\theta}})$, it has the equivalent form given by
\begin{align}
    \Tilde{\bm{\theta}} = \arg\min_{\bm{\theta} \in \bm{\Theta}} \Delta\bm{\mu}(\bm{\theta})^\text{H} \mathbf{C}_\text{a}^{-1} \Delta\bm{\mu}(\bm{\theta}).
    \label{psedo_theta_case1_appen}
\end{align}
Under the condition that $\mathbf{C}_\text{a}^{-1}$ is a scaled identity matrix, (\ref{psedo_theta_case1_appen}) leads to the result of (\ref{pseudo_para_case1}).

Next, we derive the first and second derivatives of the likelihood function $l_\text{a}(\bm{\theta})$ given by (\ref{loglike_case1}) w.r.t. the parameter vector $\bm{\theta}$.

The $m$-th entry of the first-order derivative of $l_\text{a}(\bm{\theta})$ has the following form as
\begin{align} 
    \frac{\partial l_\text{a}(\bm{\theta})}{\partial [\bm{\theta}]_m} 
    &= -\frac{\partial \bm{\epsilon}(\bm{\theta})^\text{H}}{\partial [\bm{\theta}]_m} \mathbf{C}_\text{a}^{-1} \bm{\epsilon}(\bm{\theta})
    - \bm{\epsilon}(\bm{\theta})^\text{H} \mathbf{C}_\text{a}^{-1} 
    \frac{\partial \bm{\epsilon}(\bm{\theta})}{\partial [\bm{\theta}]_m} \notag \\
    &= 2\Re \left\{ \frac{\partial \bm{\mu}(\bm{\theta})^\text{H}}{\partial [\bm{\theta}]_m} \mathbf{C}_\text{a}^{-1} \bm{\epsilon}(\bm{\theta}) \right\}.
    \label{1deriv_case1}
\end{align}
Further, the $(m,n)$-th entry of the second-order derivative can be obtained as
\begin{align} 
    \left[ \frac{\partial^2 l_\text{a}(\bm{\theta})}{\partial\bm{\theta}\partial\bm{\theta}^\text{T}} \right]_{m,n}
    =& 2\Re \left\{ 
    \frac{\partial^2 \bm{\mu}(\bm{\theta})^\text{H}}{\partial[\bm{\theta}]_m \partial[\bm{\theta}]_n} 
    \mathbf{C}_\text{a}^{-1} 
    \bm{\epsilon}(\bm{\theta}) 
    \right\} \notag\\
    &- 2\Re \left\{ 
    \frac{\partial \bm{\mu}(\bm{\theta})^\text{H}}{\partial[\bm{\theta}]_m} 
    \mathbf{C}_\text{a}^{-1} 
    \frac{\partial \bm{\mu}(\bm{\theta})}{\partial[\bm{\theta}]_n} 
    \right\}. \label{2deriv_case1}
\end{align}

Substituting (\ref{2deriv_case1}) into (\ref{J_def_gen}) yields
\begin{align} 
    [\mathbf{J}_{\Tilde{\bm{\theta}}}]_{m,n} 
    =& -2\Re \left\{ 
    \frac{\partial^2 \bm{\mu}(\bm{\theta})^\text{H}}{\partial[\bm{\theta}]_m \partial[\bm{\theta}]_n} 
    \mathbf{C}_\text{a}^{-1} 
    \mathbb{E}_p[\bm{\epsilon}(\bm{\theta})]  
    \right\} \Bigg| _{\bm{\theta}=\Tilde{\bm{\theta}}}\notag\\
    & +2\Re \left\{ 
    \frac{\partial \bm{\mu}(\bm{\theta})^\text{H}}{\partial[\bm{\theta}]_m} 
    \mathbf{C}_\text{a}^{-1} 
    \frac{\partial \bm{\mu}(\bm{\theta})}{\partial[\bm{\theta}]_n}
    \right\} \Bigg| _{\bm{\theta}=\Tilde{\bm{\theta}}},
    \label{Jmn_mid_case1}
\end{align}
where $\mathbb{E}_p[\bm{\epsilon}(\bm{\theta})] | _{\bm{\theta}=\Tilde{\bm{\theta}}} = \bm{\mu}(\bar{\bm{\theta}}) - \bm{\mu}(\Tilde{\bm{\theta}}) = \mathbf{0}$. This leads to (\ref{J_res_case1}).

Evaluated at the pseudotrue parameter $\Tilde{\bm{\theta}} = \bar{\bm{\theta}}$, (\ref{1deriv_case1}) can also be expressed as
\begin{align}
    \frac{\partial l_\text{a}(\bm{\theta})}{\partial [\bm{\theta}]_m} \Bigg| _{\bm{\theta}=\Tilde{\bm{\theta}}}
    &= 2\Re \left\{ \frac{\partial \bm{\mu}(\bm{\theta})^\text{H}}{\partial [\bm{\theta}]_m} \Bigg| _{\bm{\theta}=\Tilde{\bm{\theta}}}\mathbf{C}_\text{a}^{-1} \bm{\epsilon}(\bar{\bm{\theta}}) \right\} \notag\\
    &= 2\Re \left\{ \frac{\partial \bm{\mu}(\bm{\theta})^\text{H}}{\partial [\bm{\theta}]_m} \mathbf{C}_\text{a}^{-1} \Tilde{\bm{\epsilon}} \right\} \Bigg| _{\bm{\theta}=\Tilde{\bm{\theta}}} \notag\\
    &= (z_m+z_m^*) | _{\bm{\theta}=\Tilde{\bm{\theta}}},
\end{align}
where $z_m \triangleq \frac{\partial \bm{\mu}(\bm{\theta})^\text{H}}{\partial [\bm{\theta}]_m} \mathbf{C}_\text{a}^{-1} \Tilde{\bm{\epsilon}}$. 
Hence, the $(m,n)$-th entry of (\ref{K_def_gen}) can be obtained as
\begin{align} 
    [\mathbf{K}_{\Tilde{\bm{\theta}}}]&_{m,n} 
    = \mathbb{E}_p [(z_m+z_m^*)(z_n+z_n^*)|_{\bm{\theta}=\Tilde{\bm{\theta}}}] \notag \\
    &\overset{(a)}{=} \mathbb{E}_p [z_m z_n^*|_{\bm{\theta}=\Tilde{\bm{\theta}}}] + 
    \mathbb{E}_p [z_m^*z_n|_{\bm{\theta}=\Tilde{\bm{\theta}}}] \notag \\
    &= 2\Re\left\{ \mathbb{E}_p [z_m z_n^*|_{\bm{\theta}=\Tilde{\bm{\theta}}}] \right\} \notag\\
    &= 2\Re\left\{ 
    \frac{\partial \bm{\mu}(\bm{\theta})^\text{H}}
    {\partial [\bm{\theta}]_m} 
    \mathbf{C}_\text{a}^{-1} 
    \mathbb{E}_p[\Tilde{\bm{\epsilon}}\Tilde{\bm{\epsilon}}^\text{H}]
    \mathbf{C}_\text{a}^{-1} 
    \frac{\partial \bm{\mu}(\bm{\theta})}
    {\partial [\bm{\theta}]_n} 
    \right\}\Bigg|_{\bm{\theta}=\Tilde{\bm{\theta}}},
    \label{Kmn_mid_case1}
\end{align}
where $\mathbb{E}_p[\Tilde{\bm{\epsilon}}\Tilde{\bm{\epsilon}}^\text{H}]=\mathbf{C}(\bar{\bm{\theta}})$, leading to (\ref{K_res_case1}). In the equation above, the equality $(a)$ follows from $\mathbb{E}_p [z_m z_n]=\mathbb{E}_p [z_m^*z_n^*] = 0$, as
\begin{align}
    \mathbb{E}_p [z_mz_n] = \frac{\partial \bm{\mu}(\bm{\theta})^\text{H}}
    {\partial [\bm{\theta}]_m} 
    \mathbf{C}_\text{a}^{-1} 
    \mathbb{E}_p[\Tilde{\bm{\epsilon}}\Tilde{\bm{\epsilon}}^\text{T}]
    \mathbf{C}_\text{a}^{-1} 
    \frac{\partial \bm{\mu}(\bm{\theta})^*}
    {\partial [\bm{\theta}]_n}, \\
    \mathbb{E}_p [z_m^*z_n^*] = \frac{\partial \bm{\mu}(\bm{\theta})^\text{T}}
    {\partial [\bm{\theta}]_m} 
    \mathbf{C}_\text{a}^{-1} 
    \mathbb{E}_p[\Tilde{\bm{\epsilon}}^*\Tilde{\bm{\epsilon}}^\text{H}]
    \mathbf{C}_\text{a}^{-1} 
    \frac{\partial \bm{\mu}(\bm{\theta})}
    {\partial [\bm{\theta}]_n},
\end{align}
and the circularly symmetric complex Gaussian random variable $\Tilde{\bm{\epsilon}}$ satisfies $\mathbb{E}_p[\Tilde{\bm{\epsilon}}\Tilde{\bm{\epsilon}}^\text{T}] = \mathbb{E}_p[\Tilde{\bm{\epsilon}}^*\Tilde{\bm{\epsilon}}^\text{H}] = \mathbf{0}$.

\subsubsection{Case 2} \label{Appen_JK_case2}
Similar to Case 1, by substituting (\ref{loglike_case2}) into the objective function in (\ref{pseudo_para_case2}), we have
\begin{align}
    \arg\min_{\bm{\theta}_i \in \bm{\Theta}_i} -\mathbb{E}_p[l_\text{a}(\bm{\theta}_{i})] 
    = \arg\min_{\bm{\theta}_i \in \bm{\Theta}_i} 
    \mathbb{E}_p[\mathbf{e}(\bm{\theta}_i)^\text{H} \mathbf{C}_\text{a}^{-1} \mathbf{e}(\bm{\theta}_i)] ,
\end{align}
where $\mathbf{e}(\bm{\theta}_{i}) \triangleq \mathbf{y} - \bm{\mu}_\text{a}(\bm{\theta}_i)=\Tilde{\mathbf{e}}+\Delta\bm{\mu}_\text{a}(\bm{\theta}_i)$, $\Tilde{\mathbf{e}} \triangleq \mathbf{y} - \bm{\mu}(\bar{\bm{\theta}})$, and $\Delta\bm{\mu}_\text{a}(\bm{\theta}_i) \triangleq \bm{\mu}(\bar{\bm{\theta}}) - \bm{\mu}_\text{a}(\bm{\theta}_{i})$.
Using $\mathbb{E}_p[\mathbf{y}]=\bm{\mu}(\bar{\bm{\theta}})$, its equivalent form is given by
\begin{align}
    \Tilde{\bm{\theta}}_i = \arg\min_{\bm{\theta}_i \in \bm{\Theta}_i} \Delta\bm{\mu}_\text{a}(\bm{\theta}_i)^\text{H} \mathbf{C}_\text{a}^{-1} \Delta\bm{\mu}_\text{a}(\bm{\theta}_i),
    \label{pseudo_para_case1_append}
\end{align}
which leads to the result of (\ref{pseudo_para_case2}).

Next, we derive the first and second derivatives of the likelihood function $l_\text{a}(\bm{\theta}_i)$ given by (\ref{loglike_case2}) w.r.t. the parameter vector $\bm{\theta}_i$.

Similar to (\ref{1deriv_case1}), the $m$-th entry of the first-order derivatives has the following form as
\begin{align}
    \frac{\partial l_\text{a}(\bm{\theta}_{i})}{\partial[\bm{\theta}_{i}]_m} 
    &= -\frac{\partial \mathbf{e}(\bm{\theta}_{i})^\text{H}}{\partial[\bm{\theta}_{i}]_m} \mathbf{C}_\text{a}^{-1} \mathbf{e}(\bm{\theta}_{i})
    - \mathbf{e}(\bm{\theta}_{i})^\text{H} \mathbf{C}_\text{a}^{-1} 
    \frac{\partial \mathbf{e}(\bm{\theta}_{i})}{\partial[\bm{\theta}_{i}]_m} \notag \\
    &= 2\Re \left\{ \frac{\partial \bm{\mu}_\text{a}(\bm{\theta}_{i})^\text{H}}{\partial[\bm{\theta}_{i}]_m} \mathbf{C}_\text{a}^{-1} \mathbf{e}(\bm{\theta}_{i}) \right\},
    \label{1deriv_case2}
\end{align}

Then, the $(m,n)$-th entry of the second-order derivative is given by
\begin{align}
    \left[ \frac{\partial^2 l_\text{a}(\bm{\theta}_{i})}{\partial\bm{\theta}_{i}\partial\bm{\theta}_{i}^\text{T}} \right]_{m,n}
    =& 2\Re \left\{ 
    \frac{\partial^2 \bm{\mu}_\text{a}(\bm{\theta}_{i})^\text{H}}{\partial[\bm{\theta}_{i}]_m \partial[\bm{\theta}_{i}]_n} 
    \mathbf{C}_\text{a}^{-1} 
    \mathbf{e}(\bm{\theta}_{i}) 
    \right\} \notag\\
    &- 2\Re \left\{ 
    \frac{\partial \bm{\mu}_\text{a}(\bm{\theta}_{i})^\text{H}}{\partial[\bm{\theta}_{i}]_m} 
    \mathbf{C}_\text{a}^{-1} 
    \frac{\partial \bm{\mu}_\text{a}(\bm{\theta}_{i})}{\partial[\bm{\theta}_{i}]_n} 
    \right\}. \label{2deriv_case2}
\end{align}

Substituting (\ref{2deriv_case2}) into $[\mathbf{J}_{\Tilde{\bm{\theta}}_i}]_{m,n}$ yields
\begin{align}
    [\mathbf{J}_{\Tilde{\bm{\theta}}_i}]_{m,n} 
    =& -2\Re \left\{ 
    \frac{\partial^2 \bm{\mu}_\text{a}(\bm{\theta}_{i})^\text{H}}{\partial[\bm{\theta}_{i}]_m \partial[\bm{\theta}_{i}]_n} 
    \mathbf{C}_\text{a}^{-1} 
    \mathbb{E}_p[\mathbf{e}(\bm{\theta}_{i})]  
    \right\} \Bigg| _{\bm{\theta}_{i}=\Tilde{\bm{\theta}}_i}\notag\\
    & +2\Re \left\{ 
    \frac{\partial \bm{\mu}_\text{a}(\bm{\theta}_{i})^\text{H}}{\partial[\bm{\theta}_{i}]_m} 
    \mathbf{C}_\text{a}^{-1} 
    \frac{\partial \bm{\mu}_\text{a}(\bm{\theta}_{i})}{\partial[\bm{\theta}_{i}]_n}
    \right\} \Bigg| _{\bm{\theta}_{i}=\Tilde{\bm{\theta}}_i},
    \label{Jmn_mid}
\end{align}
where $\mathbb{E}_p[\mathbf{e}(\bm{\theta}_{i})] = \bm{\mu}(\bar{\bm{\theta}}) - \bm{\mu}_\text{a}(\bm{\theta}_{i})$. This further leads to the result of (\ref{J_res_case2}).

Meanwhile, (\ref{1deriv_case2}) can be rearranged as
\begin{align}
    \frac{\partial l_\text{a}(\bm{\theta}_{i})}{\partial [\bm{\theta}_{i}]_m} 
    =& 2\Re \left\{ \frac{\partial \bm{\mu}_\text{a}(\bm{\theta}_{i})^\text{H}}{\partial[\bm{\theta}_{i}]_m} 
    \mathbf{C}_\text{a}^{-1} 
    \Tilde{\mathbf{e}} \right\} \notag \\
    &+ 2\Re \left\{ \frac{\partial \bm{\mu}_\text{a}(\bm{\theta}_{i})^\text{H}}{\partial[\bm{\theta}_{i}]_m} 
    \mathbf{C}_\text{a}^{-1} 
    \Delta\bm{\mu}_\text{a}(\bm{\theta}_i) \right\}.
    \label{1deriv_case2_re}
\end{align}
For brevity, the first term in (\ref{1deriv_case2_re}), a random variable with expectation of $0$, is denoted as $2\Re\{z_{i,m}\}$, where the scalar $z_{i,m} = \frac{\partial \bm{\mu}_\text{a}(\bm{\theta}_{i})^\text{H}}{\partial[\bm{\theta}_{i}]_m} \mathbf{C}_\text{a}^{-1}     \Tilde{\mathbf{e}}$. 
The second term in (\ref{1deriv_case2_re}) vanishes when $\bm{\theta}_{i}=\Tilde{\bm{\theta}}_i$. This follows from the fact that $\Tilde{\bm{\theta}}_i$ is the minimizer of the optimization problem in (\ref{pseudo_para_case1_append}), which implies
\begin{align}
    &\frac{\partial \Delta\bm{\mu}_\text{a}(\bm{\theta}_i)^\text{H} \mathbf{C}_\text{a}^{-1} \Delta\bm{\mu}_\text{a}(\bm{\theta}_i)}{\partial \bm{\theta}_i}\Bigg|_{\bm{\theta}_i 
    = \Tilde{\bm{\theta}}_i} \notag\\
    &= 2\Re \left\{ \frac{\partial \bm{\mu}_\text{a}(\bm{\theta}_i)^\text{H}}{\partial \bm{\theta}_i} 
    \mathbf{C}_\text{a}^{-1} \Delta\bm{\mu}_\text{a}(\bm{\theta}_i)
    \right\}
    \Bigg|_{\bm{\theta}_i 
    = \Tilde{\bm{\theta}}_i}= \mathbf{0}.
\end{align}
Hence, the $(m,n)$-th entry of $\mathbf{K}_{\Tilde{\bm{\theta}}_i}$ can be obtained as
\begin{align}
    [\mathbf{K}_{\Tilde{\bm{\theta}}_i}]&_{m,n} 
    = \mathbb{E}_p [(z_{i,m}+z_{i,m}^*)(z_{i,n}+z_{i,n}^*)|_{\bm{\theta}_{i}=\Tilde{\bm{\theta}}_i}] \notag \\
    &\overset{(a)}{=} \mathbb{E}_p [z_{i,m}z_{i,n}^*|_{\bm{\theta}_{i}=\Tilde{\bm{\theta}}_i}] + 
    \mathbb{E}_p [z_{i,m}^*z_{i,n}|_{\bm{\theta}_{i}=\Tilde{\bm{\theta}}_i}] \notag \\
    &= 2\Re\left\{ 
    \frac{\partial \bm{\mu}_{\text{a}}(\bm{\theta}_{i})^\text{H}}
    {\partial [\bm{\theta}_{i}]_m} 
    \mathbf{C}_\text{a}^{-1} 
    \mathbb{E}_p[\Tilde{\mathbf{e}}\Tilde{\mathbf{e}}^\text{H}]
    \mathbf{C}_\text{a}^{-1} 
    \frac{\partial \bm{\mu}_{\text{a}}(\bm{\theta}_{i})}
    {\partial [\bm{\theta}_{i}]_n} 
    \right\}\Bigg|_{\bm{\theta}_{i}=\Tilde{\bm{\theta}}_i},
    \label{Kmn_mid_case2}
\end{align}
where $\mathbb{E}_p[\Tilde{\mathbf{e}}\Tilde{\mathbf{e}}^\text{H}]=\mathbf{C}(\bar{\bm{\theta}})$, leading to (\ref{K_res_case2}). Like (\ref{Kmn_mid_case1}), the equality $(a)$ of (\ref{Kmn_mid_case2}) follows from $\mathbb{E}_p [z_{i,m}z_{i,n}]=\mathbb{E}_p [z_{i,m}^*z_{i,n}^*] = 0$, as it satisfies $\mathbb{E}_p[\Tilde{\mathbf{e}}\Tilde{\mathbf{e}}^\text{T}] = \mathbb{E}_p[\Tilde{\mathbf{e}}^*\Tilde{\mathbf{e}}^\text{H}] = \mathbf{0}$.

\section{Derivatives of the Assumed Mean} \label{Appen_deriv_assume}
Following the similar procedure in Appendix \ref{Appen_deriv_true}, the derivatives of the assumed mean $\bm{\mu}_{\text{a}}(\bm{\theta}_{i})$ w.r.t. the unknown parameters $\bm{\theta}_{i}$ can be obtained. The first derivatives are given by
\begin{align}
    \frac{\partial \bm{\mu}_{\text{a}}(\bm{\theta}_{i})}{\partial \Re(h_i)} 
    &= \mathbf{B}_i \mathbf{x}_\text{p}, &
    \frac{\partial \bm{\mu}_{\text{a}}(\bm{\theta}_{i})}{\partial \Im(h_i)}
    &= j \mathbf{B}_i \mathbf{x}_\text{p},\\
    \frac{\partial \bm{\mu}_{\text{a}}(\bm{\theta}_{i})}{\partial k_i}
    &= h_i \dot{\mathbf{B}}_i^{(k)} \mathbf{x}_\text{p}, &
    \frac{\partial \bm{\mu}_{\text{a}}(\bm{\theta}_{i})}{\partial l_i}
    &= h_i \dot{\mathbf{B}}_i^{(l)} \mathbf{x}_\text{p},
\end{align}
where 
$\mathbf{B}_i\triangleq\mathbf{A} \bm{\Delta}_{k_i} \bm{\Pi}^{l_i} \mathbf{A}^\text{H}$, 
$\dot{\mathbf{B}}_i^{(k)} \triangleq -j \frac{2\pi}{N} \mathbf{A} \mathbf{D}\bm{\Delta}_{k_i} \bm{\Pi}^{l_i} \mathbf{A}^\text{H}$, 
$\dot{\mathbf{B}}_i^{(l)} \triangleq -j \frac{2\pi}{N} \mathbf{A} \bm{\Delta}_{k_i} \mathbf{F}^\text{H} \mathbf{D} \mathbf{\Delta}_{l_i} \mathbf{F} \mathbf{A}^\text{H}$, 
and $\mathbf{D}\triangleq\text{diag}(n,n=0,1,...,N-1)$.

Correspondingly, the second derivatives are presented as follows:
\begin{align}
    \frac{\partial^2 \bm{\mu}_{\text{a}}(\bm{\theta}_{i})}{\partial \Re(h_i) ^2} 
    &= \frac{\partial^2 \bm{\mu}_{\text{a}}(\bm{\theta}_{i})}{\partial \Im(h_i) ^2} 
    = \frac{\partial^2 \bm{\mu}_{\text{a}}(\bm{\theta}_{i})}{\partial \Re(h_i) \partial \Im(h_i)} 
    = \mathbf{0} \\
    \frac{\partial^2 \bm{\mu}_{\text{a}}(\bm{\theta}_{i})}{\partial \Re(h_i) \partial k_i}
    &= \dot{\mathbf{B}}_i^{(k)} \mathbf{x}_\text{p},\\
    \frac{\partial^2 \bm{\mu}_{\text{a}}(\bm{\theta}_{i})}{\partial \Re(h_i) \partial l_i}
    &= \dot{\mathbf{B}}_i^{(l)} \mathbf{x}_\text{p},\\
    \frac{\partial^2 \bm{\mu}_{\text{a}}(\bm{\theta}_{i})}{\partial \Im(h_i) \partial k_i}
    &= j \dot{\mathbf{B}}_i^{(k)} \mathbf{x}_\text{p},\\
    \frac{\partial^2 \bm{\mu}_{\text{a}}(\bm{\theta}_{i})}{\partial \Im(h_i) \partial l_i}
    &= j \dot{\mathbf{B}}_i^{(l)} \mathbf{x}_\text{p},\\
    \frac{\partial^2 \bm{\mu}_{\text{a}}(\bm{\theta}_{i})}{\partial k_i^2} 
    &= -\frac{4\pi^2}{N^2} h_i \mathbf{A} \mathbf{D}^2\bm{\Delta}_{k_i} \bm{\Pi}^{l_i} \mathbf{A}^\text{H} \mathbf{x}_\text{p},\\
     \frac{\partial^2 \bm{\mu}_{\text{a}}(\bm{\theta}_{i})}{\partial l_i^2} 
    &= -\frac{4\pi^2}{N^2} h_i \mathbf{A} \bm{\Delta}_{k_i} \mathbf{F}^\text{H} \mathbf{D}^2 \mathbf{\Delta}_{l_i} \mathbf{F} \mathbf{A}^\text{H} \mathbf{x}_\text{p},\\
     \frac{\partial^2 \bm{\mu}_{\text{a}}(\bm{\theta}_{i})}{\partial k_i \partial l_i} 
    &= -\frac{4\pi^2}{N^2} h_i \mathbf{A} \mathbf{D} \bm{\Delta}_{k_i} \mathbf{F}^\text{H} \mathbf{D} \mathbf{\Delta}_{l_i} \mathbf{F} \mathbf{A}^\text{H} \mathbf{x}_\text{p}.
\end{align}

\bibliographystyle{IEEEtran}
\bibliography{IEEEabrv, reference}

\vfill

\end{document}